\documentclass[10pt,reqno]{amsart}

\usepackage{amsmath,amssymb,mathtools}
\usepackage{latexsym}
\usepackage{epsfig}
\usepackage{cite}
\usepackage[english]{babel}
\usepackage{booktabs}
\usepackage{graphicx}
\usepackage{tikz}
\usepackage{pgfplots}
\usepackage{eurosym}
\usepackage{amsthm}
\usepackage{braket}
\usepackage{caption}
\usepackage{subcaption}
\usepackage{comment}
\newtheorem{prop}{Proposition}

\def\XXint#1#2#3{{\setbox0=\hbox{$#1{#2#3}{\int}$}
     \vcenter{\hbox{$#2#3$}}\kern-.5\wd0}}

\DeclareMathOperator{\taninv}{arctanh}

\newcommand{\bes}{\begin{equation*}}
\newcommand{\ees}{\end{equation*}}

\newcommand{\beq}{\begin{equation}}
\newcommand{\eeq}{\end{equation}}

\numberwithin{equation}{section}

\numberwithin{equation}{section}

\title{Exact solutions for analog Hawking effect in dielectric media}
\author{S. Trevisan$^{1}$, F. Belgiorno$^{4,5}$, and S.L. Cacciatori$^{2,3}$ }

\address{$^1$ Department of Information Engineering, Univerit\`a di Padova, Via Gradenigo 6/b, IT-35131 Padova, Italy. }

\address{$^2$ Department of Science and High Technology, Universit\`a dell'Insubria, Via Valleggio 11, IT-22100 Como, Italy}
\address{$^3$ INFN sezione di Milano, via Celoria 16, IT-20133 Milano, Italy}

\address{$^4$ Dipartimento di Matematica, Politecnico di Milano, Piazza Leonardo 32, IT-20133 Milano, Italy}
\address{$^5$ INdAM-GNFM}

\date{}

\begin{document}

\maketitle

\begin{abstract}
In the framework of the analog Hawking radiation for dielectric media, we analyze a toy-model 
and also the 2D reduction of the Hopfield model for a specific monotone and realistic profile for the refractive index. 
We are able to provide exact solutions, which do not require any weak dispersion approximation. The theory of Fuchsian 
ordinary differential equations is the basic tool for recovering exact solutions, which are rigoroulsy identified, and involve 
the so-called generalized hypergeometric functions $_4F_3(\alpha_1,\alpha_2,\alpha_3,\alpha_4;\beta_1,\beta_2,\beta_3;z)$. 
A complete set of connection formulas 
are available, both for the subcritical case and for the transcritical one, and also the Stokes phenomenon 
occurring in the problem is fully discussed. From the physical point of view, we focus on the problem of thermality. Under suitable 
conditions, the Hawking temperature is deduced, and we show that it is in fully agreement with the expression deduced in other frameworks 
under various approximations. 
\end{abstract}


\section{Introduction}

We are interested to focus our attention on analytical calculation of the analog Hawking effect in dielectric media, and in presence of dispersion. 
Analytical calculations for the analog Hawking effect, introduced in the seminal paper \cite{unruh} for non-dispersive media, also in the dispersive case have been largely discussed in literature, see e.g. the following (non-exhaustive) list of papers~\cite{brout,corley,himemoto,saida,Schutzhold-Unruh,balbinot,unruh-s,cpf,Leonhardt-Robertson,Coutant-prd,Coutant-und,Un-schu,petev,Coutant-thick,hopfield-hawking,linder,philbin-exact,hopfield-kerr,CoutantWeinfurtner,coutant-kdv,coutant-bdg,Konig}, and for weak dispersion and the 
transcritical case a rather general mathematical framework, able to encompass in an unified picture very relevant models even for the experiments 
\cite{rousseaux-first,faccio-prl,rubino-njp,weinfurtner-prl,rousseaux-book,weinfurtner-book,jeff-nature,rousseaux,denova,drori,rousseaux-cocurrent}, has been discussed in 
\cite{master,bec,belcaccia-universe}. 
In \cite{PaperoVariazionale}, the authors introduced a new mathematical perspective in the analog Hawking effect, by relating the problem to the solution of a fourth-order Fuchsian equation for the subcritical case. 
As a remarkable example of the possibilities offered by Fuchsian equations, we provide here an \emph{exact} solution of a particular scattering problem inside a dielectric. 
We choose a monotonic refractive index profile traveling inside the dielectric at a fixed velocity. This kind of background 
represents a fundamental setting for a good understanding of the pair-creation process, and monotonic backgrounds are often considered in analytical computations in transcritical regime and in numerical simulations of analog systems ( see \emph{e.g.} \cite{Parentani_numerical}). We stress that exact solutions for the dispersive case of the analog Hawking effect are very hard to be obtained, as even solutions of ordinary differential equations of the fourth order. The only other example we can find is contained in the paper by Philbin \cite{philbin-exact}, which provides exact solutions for the Corley model \cite{corley,cpf}. Therein solutions are obtained at the price of introducing an interesting but somehow unrealistic linear velocity profile $v(x)=-\alpha x$.\\

The plan of the paper is the following. In Section \ref{sec_Exact_background} 
we present the model and the monotonic background we are going to consider. In Section \ref{sec_FuchsianEqPsi}, starting from the equation of motion for our model, we show 
that, by means of a suitable change of the independent spatial variable in the comoving frame of the pulse, we are able to obtain a fourth order Fuchsian equation. We provide a detailed characterization of its local monodromy and spectral type. In Section \ref{sec_ExactSolution} we provide the exact solution, and recover rigorously that the 
generalized hypergeometric functions $_4F_3(\alpha_1,\alpha_2,\alpha_3,\alpha_4;\beta_1,\beta_2,\beta_3;z)$ are involved. Furthermore, we provide a study of the Stokes 
phenomenon, and we study some physical consequences for the scattering problem at hand. In Section \ref{original-fipsi}, we consider a generalization of the previous analysis to the case of the original model (the so-called $\phi\psi$-model) to which the Hopfield model reduces in the 2D case, and again we consider the scattering problem and the thermality of the spectrum, which is recovered to coincide with the one deduced in the weak dispersion limit discussed in \cite{master} under suitable conditions. In Section \ref{conclusions}, we 
summarize our achievements and display future perpectives for our analysis.

\section{The Cauchy model and the choice of background}
\label{sec_Exact_background}

The model we consider is the modified $\phi$-$\psi$ model (or ``Cauchy model'') introduced by the authors in \cite{PaperoVariazionale}. In the laboratory frame, it is expressed by the Lagrangian
\begin{align}\label{Eq_Lagrangian}
\mathcal{L} = \frac 1 2 (\partial_{t_l}\phi)^2 + \frac 1 2 \left((\partial_{t_l} \psi)^2 + \mu^2\psi^2\right) + g \phi \partial_{x_l} \psi - \frac \lambda {4!} \psi^4 \,.
\end{align}
The linearized EOMs around a background solution $\psi_B$, in the lab frame, are
\begin{align} \label{EOM_LinCauchy}
&\partial_{t_l}^2 \phi- g \partial_{x_l}\psi=0,\\
&\partial_{t_l}^2 \psi+ g \partial_{x_l} \psi-\mu^2 \phi+\frac{\lambda}{2} \psi_B^2 \psi=0. \label{EOM_LinCauchy2}
\end{align}
The free-field solutions (for $\lambda=0$) are plane waves $e^{i\omega_{lab} t_l - i k_{lab} x_l}$ which satisfy
\begin{align}\label{CauchyDisp}
n_0^2(\omega_{lab}):=\frac{k_{lab}^2}{\omega_{lab}^2} = \frac {\mu^2}{g^2} + \frac{\omega_{lab}^2}{g^2} =: A + B\omega_{lab}^2 \,.
\end{align}
The dispersion relation, in a reference frame moving with velocity $V$ with respect to the lab, has four solutions (see Figure \ref{fig_DR}): expanding $k(\omega)$ for $\omega\rightarrow 0$, the four modes have the following expressions
{\small
	\begin{align}\label{Eq_kH_lowW}
	k_H &= \frac{\mu - g V}{g - \mu V}\,\omega + O(\omega^3), \\
	k_B &= -\frac{\mu + g V}{g + \mu V}\,\omega + O(\omega^3), \\
	k_P &=  \frac{\sqrt{g^2 - \mu^2 V^2}}{
		\gamma V^2 } - \left(\frac 1 V + \frac {g^2} {\gamma^2 V (g^2 - \mu^2 V^2)}\right) \,\omega  - \frac{g^2 (2 g^2 + \mu^2 V^2) }{2\gamma  (g^2 - \mu^2 V^2)^{5/2} }\,\omega^2 + O(\omega^3) , \\
	k_N &=  -\frac{\sqrt{g^2 - \mu^2 V^2}}{
		\gamma V^2 } - \left(\frac 1 V + \frac {g^2} {\gamma^2 V (g^2 - \mu^2 V^2)}\right) \,\omega  + \frac{g^2 (2 g^2 + \mu^2 V^2) }{2\gamma  (g^2 - \mu^2 V^2)^{5/2} }\,\omega^2 + O(\omega^3)\,.\label{Eq_kN_lowW}
	\end{align}}
\begin{figure}
	\centering
	\includegraphics[scale=0.45]{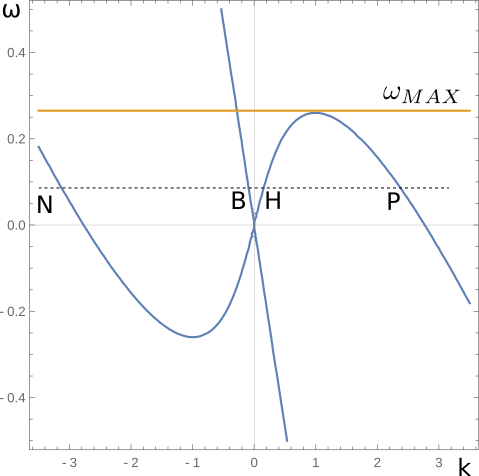}
	\caption{  The dispersion relation (\ref{CauchyDisp}) represented in the comoving frame with the background, with $g=1$, $\mu=1.2$; for $0<\omega<\omega_{MAX}$ there are four real solutions. }
	\label{fig_DR}
\end{figure}

The reason of the choice of the model lies in  the simple expression of the dispersion relation: however, in many cases we will refer to $\text{DR}(k)$ as a generic fourth order polynomial, so the choice of a particular dispersion relation is not really crucial. We are going to solve the linearized equations (\ref{EOM_LinCauchy})-(\ref{EOM_LinCauchy2}) with a particular choice of  background. 
In the experiments the background field is represented by a laser pulse, which is naturally \emph{localized} and travels rigidly at a certain velocity $V$. We will instead consider a monotonic background 
\begin{align}
\psi_b(x-Vt)^2 = 1 - \tanh(\beta (x-V t)) \,.
\end{align} 
We claim that this is a good model for the right-side of a laser pulse; moreover, a monotonic background represents a better model for an event horizon and allows to better understand the nature of Hawking Radiation. Monotonic backgrounds of this type were also used in some previous studies of analog black holes and white holes (see for example \cite{ Parentani_numerical}). 

The linearized equations (\ref{EOM_LinCauchy}) and (\ref{EOM_LinCauchy2}) can be put together to a fourth-order equation of generalized Orr-Sommerfeld type. 
It is convenient to write these equations in the comoving coordinates $t=\gamma (t_l - V x_l)$, $x=\gamma (x_l - V t_l)$. Since the potential term is independent of the comoving time, we seek for a solution in the form $\psi = e^{-i\omega t} f(x)$, $\phi = e^{-i\omega t} g(x)$. By applying $(v^\mu\partial_\mu)^2$ to the second equation, we obtain a single fourth order equation for $f(x)$ only:
\begin{align}\label{Eq_FourthOrder}
&0 = V^4 \gamma^4 f^{(4)}(x) +4 i V^3 \omega \gamma^4 f^{(3)}(x) + \cr 
&+\frac{1}{2} \gamma^2
f''(x) \left(-\lambda V^2 \tanh (\tilde\beta x)+2 g^2+\lambda V^2-2 \mu^2 V^2-12 V^2 \omega^2 \gamma^2\right)+ \cr 
&+  i\gamma^2 V f'(x) \left(i \tilde\beta \lambda V \text{sech}^2(\tilde \beta x)+ \omega \left(-\lambda \tanh (\tilde \beta x)+2 g^2+\lambda-2 \mu^2-4 \omega^2 \gamma^2\right)\right) + \cr
&+ \frac{1}{2} \gamma^2 f(x) \left[\omega^2 \left(\lambda \tanh (\tilde \beta x)-2 g^2 V^2-\lambda+2
\mu^2+2 \omega^2 \gamma^2\right)\right.\cr 
&\left. +2 \tilde \beta \lambda V \text{sech}^2(\tilde \beta x) (\tilde \beta V \tanh (\tilde \beta x)-i \omega)\right] \,,
\end{align}
where $\tilde \beta=\frac{\beta} \gamma$.

\section{Reduction to a Fuchsian equation: monodromy and Riemann scheme}\label{sec_FuchsianEqPsi}

As in \cite{PaperoVariazionale}, we perform the
following change of variables on Eq. (\ref{Eq_FourthOrder}),
\begin{align}
z&=-e^{2\tilde \beta x}\,,
\end{align} 
which implies
\begin{align}
\partial_x &= 2\tilde \beta \theta_z := 2\tilde \beta z \frac d {dz}\,.
\end{align}
By defining the rescaled parameters $G=\frac g {2\tilde \beta}$, $\Omega=\frac \omega {2\tilde \beta}$, $M=\frac {\mu} {2\tilde \beta}$ and $\Lambda=\frac {\lambda} {4\tilde \beta^2}$, we end up with the following equation:
\begin{align}\label{Eq_Monotonic_FuchsianEQ}
&0 = V^4 \gamma^4 z^4 f^{(4)}(z) +f^{(3)}(z) \left(6 V^4 \gamma^4 z^3-4 i V^3 \Omega \gamma^4
z^3\right) \cr  
&+f''(z) \left(G^2 \gamma^2 z^2-V^2 \gamma^2 z^2
\left(\frac{\Lambda}{z-1}+M^2\right) \right. \cr 
& \left. \qquad-6 V^2 \Omega^2 \gamma^4 z^2-12 i V^3 \Omega \gamma^4 z^2+7 V^4 \gamma^4 z^2\right) \cr
&+f'(z) \left(-2 i G^2 V \Omega \gamma^2 z+G^2 \gamma^2 z+V^2 \gamma^2 z \left(\frac{\Lambda}{(1-z)^2}-M^2\right)-2 i V \Omega \gamma^2 z \left(\frac{\Lambda}{(1-z)}-M^2\right) \right. \cr 
&\left. \qquad +\frac{2 \Lambda V^2
	\gamma^2 z^2}{(1-z)^2}-6 V^2 \Omega^2 \gamma^4 z-4 i V^3 \Omega \gamma^4 z+V^4 \gamma^4 z+4 i V \Omega^3 \gamma^4 z\right)
\cr 
&+f(z) \left(-G^2 V^2 \Omega^2 \gamma^2+\Omega^2 \gamma^2 \left(\frac{\Lambda}{1-z}+M^2\right) \right. \cr 
& \left. \qquad +\frac{2 \Lambda V^2 \gamma^2 z^2}{(1-z)^3}+\frac{\Lambda
	V^2 \gamma^2 z}{(1-z)^3}-\frac{2 i \Lambda V \Omega \gamma^2 z}{(1-z)^2}+\Omega^4 \gamma^4\right) \,.
\end{align} 
This equation is of Fuchsian type, with three singular points, $z=0,1,\infty$.

As an alternative form, we can write the EOMs (\ref{EOM_LinCauchy})-(\ref{EOM_LinCauchy2}) as a system of first order and perform the same change of variables as before. With this procedure we obtain the system
\begin{align}\label{Eq_Monotonic_system}
\frac{dU}{dz}= A(z) U\,,
\end{align}
were
\begin{align}
U &= \begin{pmatrix}
g(z) \\g'(z) \\f(z) \\f'(z)
\end{pmatrix}\\
A(z)&=\left(
\begin{array}{cccc}
0 & 1 & 0 & 0 \\
\frac{\Omega^2}{V^2 z^2} & -\frac{V^2 \gamma^2 z-2 i V \Omega \gamma^2 z}{V^2 \gamma^2 z^2} & \frac{i G \Omega}{V \gamma z^2} & -\frac{G}{V^2 \gamma z} \\
0 & 0 & 0 & 1 \\
-\frac{i G \Omega}{V \gamma z^2} & \frac{G}{V^2 \gamma z} & -\frac{-\frac{\Lambda}{1-z}-M^2-\Omega^2 \gamma^2}{V^2 \gamma^2 z^2} & -\frac{V^2 \gamma^2 z-2 i V \Omega \gamma^2 z}{V^2 \gamma^2 z^2} \\
\end{array}
\right)\,.
\end{align}
We can reduce (\ref{Eq_Monotonic_system}) to a \emph{Fuchsian System of Normal Form} \cite{Haraoka} by changing variable to 
\begin{align}
Y(z) &= P(z)U(z)\,, \\
P(z) &= \text{diag}\left(\frac 1 z , 1 , \frac 1 z , 1 \right)\,.
\end{align}
The system has now the form
\begin{align}
\frac{dY}{dz}&=\left(\frac{A_1}{z} + \frac{A_2}{z-1}\right)Y \,,\\
A_1 &= \left(
\begin{array}{cccc}
-1 & 1 & 0 & 0 \\
\frac{\Omega^2}{V^2} & \frac{-V-2 i \Omega}{V} & -\frac{i G \Omega}{V \gamma} & -\frac{G}{V^2 \gamma} \\
0 & 0 & -1 & 1 \\
\frac{i G \Omega}{V \gamma} & \frac{G}{V^2 \gamma} & \frac{\Lambda+M^2+\Omega^2 \gamma^2}{V^2 \gamma^2} & \frac{-V-2 i \Omega}{V} \\
\end{array}
\right)\,, \\
A_2 &= \left(
\begin{array}{cccc}
0 & 0 & 0 & 0 \\
0 & 0 & 0 & 0 \\
0 & 0 & 0 & 0 \\
0 & 0 & -\frac{\Lambda}{V^2 \gamma^2} & 0 \\
\end{array}
\right) \,.
\end{align}
The matrices $A_1$ and $A_2$ are constant and they are, respectively, the residue at the simple poles $z=0$ and $z=1$. We may also define
\begin{align}
A_0:=-A_1 - A_2 \,,
\end{align}
which corresponds to the residue at the simple pole $z=\infty$. 

\subsection{The local solutions and monodromy}
\label{sec_Monotonic_localsol}

We start looking for local solutions of (\ref{Eq_Monotonic_FuchsianEQ}) around $z=\infty$. After changing variable to $t=1/z$, we can look for a solution in the form 
\begin{align}\label{Eq_FormalSolution_inf}
f(t) = t^{-i \alpha} \sum_{n=0}^{\infty} c_n t^n \,.
\end{align}
The characteristic equation for the exponent $k:=2\tilde \beta \alpha$ is
\begin{align}\label{Eq_CharEq_inf}
\text{DR}(k)&:= 
\gamma^2 (\mu^2 (k V + \omega)^2 - g^2 (k + V \omega)^2 + (k V + \omega)^4 \gamma^2)=0\,,
\end{align}
which is nothing but the dispersion relation (\ref{CauchyDisp}) as written in the comoving reference frame with the background. Equation (\ref{Eq_CharEq_inf}) as in general four distinct complex solutions, so we find four independent solutions in the form (\ref{Eq_FormalSolution_inf}): the spectral type of the equation at $z=\infty$ is thus $(1111)$.  
There is the  possibility of  emergence of a resonant case, where the difference between eigenvalues is an integer, 
still in a zero measure set in the space of available parameters appearing in our model. 

A similar behaviour is found at $z=0$: we find four independent local solutions of the form
\begin{align}\label{Eq_FormalSolution_zero}
f(z) = z^{i \tilde\alpha} \sum_{n=0}^{\infty} c_n z^n \,,
\end{align}
where $\tilde k :=2\tilde \beta \tilde\alpha$ satisfies 
\begin{align}\label{Eq_CharEq_zero}
\text{DR}_0(k) & :=
\gamma^2 (-g^2 (\tilde k + V \omega)^2 + (\tilde k V + \omega)^2 (\lambda + \mu^2 + (\tilde k V + \omega)^2 \gamma^2))=0 \,.
\end{align}
Eq. (\ref{Eq_CharEq_zero}) is equivalent to (\ref{Eq_CharEq_inf}) if one maps $\mu^2 \mapsto \mu^2 + \lambda$. The spectral type at $z=0$ is again $(1111)$,  again almost everywhere in the space of available parameters appearing in our model. Some interesting formulas concerning \eqref{Eq_CharEq_inf} and \eqref{Eq_CharEq_zero} are discussed in Appendix \ref{useful-relations}.

The situation at $z=1$ is different. After defining $y=z-1$, the characteristic equation for solutions of the form
\begin{align}\label{Eq_FormalSolution_1}
f(y) = y^{a} \sum_{n=0}^{\infty} c_n y^n 
\end{align}
has four integer solutions $a=0,1,2,3$.  This situation is known in literature as the resonant case (cf. e.g. \cite{Haraoka}), and requires 
a particular study. We refer mostly to \cite{coddington}, where still the discussion is left incomplete, and, particularly, to the thoroughful analysis appearing 
in \cite{forsythe4}, and also to \cite{craig}. As suggested in the aforementioned literature, we apply the so-called Frobenius method for the analysis 
of solutions at a Fuchsian singularity, also in the resonant case, and we can also verify if there are logarithmic contributions (even in the resonant case, 
they might also not appear).

By means of the Frobenius method, we obtain three independent integer solutions
\begin{align}\label{Eq_localsol_1}
u_1(y) &= y^3+y^4 \left[ -\frac{3}{2}+\frac{\Lambda}{12 \gamma^2 V^2}+i\frac{\Omega}{V}\right]  + o(y^4)\\
u_2(y) &= y^2 +y^3 \left[ -2+\frac{\Lambda}{6 \gamma^2 V^2}+i\frac{4\Omega}{3V}\right] + o(y^3) \\
u_3(y)&= y\; +y^2 \left[-3+\frac{\Lambda}{2 \gamma^2 V^2}+i\frac{2\Omega}{V}\right] + o(y^2) \,,
\end{align}
and one logarithmic solution 
\begin{align}\label{Eq_localsol_4}
&u_0(y) =1+y \left[ -6-\frac{\Lambda}{V^2 \gamma^2}+\frac{4 i \Omega}{V}\right]  +o(y)\cr 
& +\log (y) \Big(R_1 u_1(y)+R_2 u_2(y)+R_3 u_3(y)\Big) \,,
\end{align} 

where  
	\begin{align}
	R_3 &= \frac{\Lambda}{V^2 \gamma^2} \\
	R_2 &= \frac{\Lambda \left(5 V-2 i \Omega\right)}{4 V^3 \gamma^2} \\
	R_1 &= \frac{\Lambda \left(9 V^2-6 i \Omega V-\Omega^2\right)}{18 V^4 \gamma^2} \,.
	\end{align}
	
The study  of the monodromy of the solution is important for the characterization of the equation \cite{Haraoka}. Starting from a basis of solutions $(u_1(\bar z),...,u_4(\bar z))$ evaluated at some $\bar z\in \mathbb C$, we can prolong these solutions along a path that goes around a singular point $a \in \mathbb C$ and closes back to $\bar z$ (without enclosing other singular points): the new vector $(u'_1(\bar z),...,u'_4(\bar z))$ that results from this transformation  is related to the initial one by a matrix $M_a$. Such matrix is independent on the point $\bar z$ and is called the \emph{monodromy} matrix of the solutions $(u_1,...,u_4)$ at the point $a$. 
The monodromy matrix of the solutions $(u_0(y),u_1(y),u_2(y),u_3(y))$ of (\ref{Eq_localsol_1})-(\ref{Eq_localsol_4}) at $z=1$ is easily computed as
\begin{align}
M_1 = \left(
\begin{array}{cccc}
1 & 2 \pi i  R_1 & 2  \pi i R_2 & 2 \pi i R_3 \\
0 & 1 & 0 & 0 \\
0 & 0 & 1 & 0 \\
0 & 0 & 0 & 1 \\
\end{array}
\right) \,,
\end{align}
whose Jordan form is 
\begin{align}
J_{M_1} = \left(
\begin{array}{cccc}
1 & 0 & 0 & 0 \\
0 & 1 & 0 & 0 \\
0 & 0 & 1 & 1 \\
0 & 0 & 0 & 1 \\
\end{array}
\right)\,.
\end{align}
The monodromy at $z=0$ and $z=\infty$ are even more easy to determine and they are represented respectively by the diagonal matrices
\begin{align}
M_0 &= \text{diag}\left(e^{i2\pi \tilde \alpha_1},e^{i2\pi \tilde \alpha_2},e^{i2\pi \tilde \alpha_3},e^{i2\pi \tilde \alpha_4}\right)\,,\\
M_\infty &= \text{diag}\left(e^{-i2\pi  \alpha_1},e^{-i2\pi  \alpha_2},e^{-i2\pi \alpha_3},e^{-i2\pi \alpha_4}\right)\,,
\end{align}
whose Jordan form is $J_{M_0}=J_{M_\infty}= \mathbb I_4$.
From the Jordan form we can infer that the spectral type of the equation at $z=1$ is $(3,1)$ (see \cite{Haraoka}). 
The spectral type of the equation is thus $((1111),(31),(1111))$: this spectral type is classified as \emph{rigid}. Without entering into mathematical details, an equation is called rigid if the local monodromy class of its solutions uniquely determines also the \emph{global} monodromy class. Another way of expressing the same concept is that the equation only depends on its \emph{local data} (\emph{i.e.} the characteristic exponents) and there are no \emph{accessory parameters}. 
We can also calculate the so-called \emph{index of rigidity} \cite{Haraoka}:
\begin{align}
\iota = \sum_{j=0}^p \text{dim}Z(M_j) - (p-1)n^2 \,,
\end{align}
where $Z(M_j)$ is the centralizer of the matrix $M_j$ (\emph{i.e.} the dimension of the vector space of matrices that commute with $M_j$) and $p+1$ is the number of distinct singular points of the equation.  A known result says that a Fuchsian system is rigid if and only if $\iota=2$. In our case, a simple computation leads to
\begin{align}
Z(M_0)+Z(M_1)+Z(M_\infty)-(4)^2 = 4+10+4-16 = 2\,.
\end{align}

Rigid equations have thus a simple structure and there are many results available for their characterization and solution. 
The rigidity of the equation allows in principle to find integral representations of the solutions and write exact expressions of connection coefficients for the local solutions at the different singular points: for the physical problem of the scattering. 
Another consequence of rigidity is that the local monodromy classes uniquely determine the global monodromy: this is interesting for physics, since the action of the global monodromy can be interpreted as the result of the scattering of a wave, so the scattering coefficients may be derived from the monodromy matrices \cite{Castro,DaCunha}.

\subsection{Gauge Transformation and Riemann scheme}

We start again from Eq. (\ref{Eq_Monotonic_FuchsianEQ})  and we perform a so-called gauge transformation in order to put to zero one of the characteristic exponents. We look for a solution of the form
\begin{align}
f(z):= z^{i \frac {\tilde  k_1}{2 \tilde \beta}}(z-1)\, u(z)\,,
\end{align}
where $\tilde k_1$ satisfies $$\text{DR}_0(\tilde k_1)=0\,.$$ The exponent of $(z-1)$ was chosen to lower the order of the singularities at $z=1$.  The function $u(z)$ now satisfies 
{\footnotesize
	\begin{align}\label{Eq_FuchsianEQ_gauge2}
	u(z)& \Big[\gamma^2 \left(G^2 (\tilde k_1+V \Omega)^2-(\tilde k_1 V+\Omega)^2 \left(\gamma^2 (\tilde k_1
	V+\Omega)^2+\Lambda+M^2\right)\right)\cr &+\gamma^2 z \left(-G^2 (\tilde k_1+V \Omega-i)^2+M^2 (\Omega+(\tilde k_1-i) V)^2+\gamma^2 (\Omega+(\tilde k_1-i)
	V)^4\right)\Big] + \cr
	u'(z)& \Big[\gamma^2 z \Big(V (2 i k_1 V+V+2 i \Omega) \left(-i \gamma^2 ((1+i) k_1 V-i V+(1+i) \Omega) ((1+i) k_1
	V+V+(1+i) \Omega)+L+M^2\right)\cr 
	&\qquad-i G^2 (2 k_1+2 V \Omega-i)\Big)\cr 
	&+\gamma^2
	z^2 \Big(-i V (2 \Omega+(2 k_1-3 i) V)
	\left(M^2+\gamma^2 \left((-5+2 k_1 (k_1-3 i)) V^2+2 (2 k_1-3 i) V \Omega+2 \Omega^2\right)\right)\cr 
	&\qquad 2 i G^2 k_1+2 i G^2 V \Omega+3 G^2\Big)\Big]+ \cr 
	u''(z) &\Big[\gamma^2 z^2 \left(-G^2+6 \tilde k_1^2 V^4 \gamma^2+12 \tilde k_1 V^3 \Omega \gamma^2-12 i \tilde k_1 V^4 \gamma^2+\Lambda V^2+M^2 V^2+6 V^2 \Omega^2
	\gamma^2-12 i V^3 \Omega \gamma^2-7 V^4 \gamma^2\right)\cr & +\gamma^2 z^3 \left(G^2-6 \tilde k_1^2 V^4 \gamma^2-12 (\tilde k_1-2 i) V^3 \Omega \gamma^2+24 i \tilde k_1
	V^4 \gamma^2-M^2 V^2-6 V^2 \Omega^2 \gamma^2+25 V^4 \gamma^2\right)\Big]+ \cr 
	u^{(3)}(z) &\Big[ \gamma^2 z^4 \left(4 i \tilde k_1 V^4
	\gamma^2+4 i V^3 \Omega \gamma^2+10 V^4 \gamma^2\right)+\gamma^2 z^3 \left(-4 i \tilde k_1 V^4 \gamma^2-4 i V^3 \Omega \gamma^2-6 V^4 \gamma^2\right)\Big]+ \cr 
	u^{(4)}(z) & z^4 (z-1) V^4 \gamma^4 =0 \,.
	\end{align}  
}

The last equation can be written in a more convenient form using (\ref{Eq_CharEq_inf}) and (\ref{Eq_CharEq_zero}):
{\small 
	\begin{align}\label{Eq_FuchsianEQ_DR}
	u(z)& \big[-\text{DR}_0(\tilde k_1)+z \text{DR}(\tilde k_1-i)\big] + \cr 
	u'(z)& \Big[z (\text{DR}_0(\tilde k_1)-\text{DR}_0(\tilde k_1-i))-z^2 (\text{DR}(\tilde k_1-i)-\text{DR}(\tilde k_1-2
	i))\Big] +\cr 
	u''(z) & \Big[ \frac{1}{2} z^3 (\text{DR}(\tilde k_1-i)-2 \text{DR}(\tilde k_1-2 i)+\text{DR}(\tilde k_1-3
	i)) \cr 
	&-\frac{1}{2} z^2 (\text{DR}_0(\tilde k_1)-2 \text{DR}_0(\tilde k_1-i)+\text{DR}_0(\tilde k_1-2 i))\Big] +\cr 
	u^{(3)}(z) &
	\Big[\frac{1}{6} z^3 (\text{DR}_0(\tilde k_1)-3 \text{DR}_0(\tilde k_1-i)+3 \text{DR}_0(\tilde k_1-2 i)-\text{DR}_0(\tilde k_1-3
	i))\cr 
	&-\frac{1}{6} z^4 (\text{DR}(\tilde k_1-i)-3 \text{DR}(\tilde k_1-2 i)+3 \text{DR}(\tilde k_1-3 i)-\text{DR}(\tilde k_1-4
	i))\Big] +\cr 
	u^{(4)}(z) & \Big[ \frac{1}{24} z^5 (\text{DR}(\tilde k_1-i)-4 \text{DR}(\tilde k_1-2 i)+6 \text{DR}(\tilde k_1-3
	i)-4 \text{DR}(\tilde k_1-4 i)+\text{DR}(\tilde k_1-5 i))\cr 
	-\frac{1}{24}& z^4 (\text{DR}_0(\tilde k_1)-4
	\text{DR}_0(\tilde k_1-i)+6 \text{DR}_0(\tilde k_1-2 i)-4 \text{DR}_0(\tilde k_1-3 i)+\text{DR}_0(\tilde k_1-4 i))\Big] =0\,.\cr
	\end{align}}
where $\text{DR}_0(\tilde k_j)=0$ and $\text{DR}( k_j)=0$.

It is easy to verify, by studying the local solutions as in Section \ref{sec_Monotonic_localsol}, that the characteristic exponents of Eq.~(\ref{Eq_FuchsianEQ_DR}) are 
\begin{align}\label{Eq_Monotonic_RiemannScheme}
\left[\begin{matrix}
z=0 & z=1 & z=\infty \\
0	& 0	& 1-\frac i {2\tilde \beta}( k_1 -  \tilde k_1)\\
\frac i {2\tilde \beta}(\tilde k_2 -  \tilde k_1)	& 1	& 1-\frac i {2\tilde \beta}( k_2 - \tilde k_1) \\
\frac i {2\tilde \beta}(\tilde k_3 - \tilde k_1)	& 2	& 1- \frac i {2\tilde \beta}(k_3 - \tilde k_1) \\
\frac i {2\tilde \beta}(\tilde k_4 - \tilde  k_1) & -1	& 1-\frac i {2\tilde \beta}(k_4 -\tilde k_1)
\end{matrix}\right]\,.
\end{align}
Eq. (\ref{Eq_Monotonic_RiemannScheme}) is the so-called Riemann Scheme of the equation\footnote{More precisely, the Riemann P-Scheme is usually written as
\begin{align}
\mathcal P\left\{\begin{matrix}
w' & w'' & w''' \\
a_1	& b_1	& c_1\\
a_2	& b_2	& c_2\\
a_3	& b_3	& c_3\\
a_4	& b_4	& c_4
\end{matrix}\ ; z\right\}\,,
\end{align}
and indicates independently the equations and the solutions. It is often also called Papperitz symbol, but it has been introduced by Riemann first.}.
By defining
\begin{align}\label{Eq_hypergeom_alpha}
\alpha_i &:= 1-\frac i {2 \tilde\beta}( k_i - \tilde k_1) \,,\quad i=1,2,3,4\\ \label{Eq_hypergeom_beta}
 \beta_j &:= 1 - \frac i {2 \tilde \beta}(\tilde k_{j+1} - \tilde k_1)\,,\quad j=1,2,3\,,
\end{align}
we can write Eq. (\ref{Eq_Monotonic_RiemannScheme}) as
\begin{align}
\left[\begin{matrix}
z=0 & z=1 & z=\infty \\
0	& 0	& \alpha_1 \\
1- \beta_1	& 1	& \alpha_2 \\
1- \beta_2	& 2	& \alpha_3 \\
1-\beta_3 & - \beta_4	& \alpha_4
\end{matrix}\right]\,.
\end{align}
which corresponds to  the Riemann scheme of the hypergeometric function $$_4F_3(\alpha_1,\alpha_2,\alpha_3,\alpha_4;\beta_1,\beta_2,\beta_3;z)$$ in the standard form \cite{Oshima}. The exponent $\beta_4$ in the hypergeometric function is defined by $\sum_{i=0}^4\alpha_i=\sum_{i=0}^4\beta_i$, and is indeed equal to 1.
Therefore, the spectral type and the Riemann scheme of our fourth order equation coincide with those of the hypergeometric function $_4F_3$. Since the system is rigid, Eq. (\ref{Eq_FuchsianEQ_gauge2}) has to be equivalent to the hypergeometric equation $_4E_3$ \cite{aomoto}, and $_4 F_3$ has to be a solution as we are now going to see.

\section{The exact solution: hypergeometric  $ \,_4F_3$, Stokes phenomenon and connection formulas}
\label{sec_ExactSolution}

We look for a locally holomorphic solution of Eq. (\ref{Eq_FuchsianEQ_DR}) around $z=0$
\begin{align}\label{Eq_SolutionExpansion}
u(z)=1 +c_1 z + c_2 z^2 + ...
\end{align}
We are going to prove the following proposition that gives the explicit expression of the coefficients $c_n$:

\begin{prop}
	\label{Th_Cn} Given any two fourth order polynomials $\text{DR}(k)$ and $\text{DR}_0(k)$, let $\tilde k_1$ be one of the roots of $\text{DR}_0$. Let $u(z)$ be a meromorphic function  which solves Eq. (\ref{Eq_FuchsianEQ_DR}) and suppose that $u(z)$ is locally holomorphic around $z=0$. Then, the general term of the series expansion (\ref{Eq_SolutionExpansion}) satisfies
	\begin{align}\label{Eq_Cn}
	c_n = \frac{\prod_{r=1}^{n}\text{DR}(\tilde k_1-r\,i )}{\prod_{s=1}^{n}\text{DR}_0(\tilde k_1-s \,i)}\,.
	\end{align}	
\end{prop} 

\begin{proof}
See Appendix \ref{proposition-proof}.
	
\end{proof}

Using the definitions (\ref{Eq_hypergeom_alpha})-(\ref{Eq_hypergeom_beta})
and writing the dispersion relations in terms of their roots as in (\ref{Eq_DR_prod}) and (\ref{Eq_DR0_prod}), we easily find
\begin{align}\label{Eq_SeriesCoeff}
c_n = \frac{\prod_{i=1}^{4}\alpha_i(\alpha_i+1)...(\alpha_i+n) }{n!\prod_{j=1}^{3} \beta_j(\beta_j+1)...(\beta_j+n)} \,.
\end{align}
This is precisely the general term of the hypergeometric function $\,_4F_3$. So we can say that 
\begin{align}
u(z)=  \,_4F_3(\alpha_1,\alpha_2,\alpha_3,\alpha_4;\beta_1,\beta_2,\beta_3;z)
\end{align}
is an exact solution of (\ref{Eq_FuchsianEQ_gauge2}), and 
\begin{align}\label{Eq_ExatSolution}
f(z)= z^{i\frac{\tilde k_1}{2\tilde \beta}}(z-1) \,_4F_3(\alpha_1,\alpha_2,\alpha_3,\alpha_4;\beta_1,\beta_2,\beta_3;z)
\end{align}
is a solution of (\ref{Eq_Monotonic_FuchsianEQ}).
Solving the scattering problems now just amounts to writing the connection coefficients of the hypergeometric function between $z=0$ and $z=\infty$: for example, the connection coefficient $\tilde k_1 \rightarrow k_1$ is
\begin{align}\label{Eq_ConnCoeff}
C_{\tilde k_1\rightarrow k_1}=\frac{\Gamma (\beta_1) \Gamma (\beta_2) \Gamma (\beta_3) \Gamma (\alpha_2-\alpha_1) \Gamma (\alpha_3-\alpha_1) \Gamma
	(\alpha_4-\alpha_1)}{\Gamma (\alpha_2) \Gamma (\alpha_3) \Gamma (\alpha_4) \Gamma (\beta_1-\alpha_1) \Gamma (\beta_2-\alpha_1)
	\Gamma (\beta_3-\alpha_1)}\,.
\end{align}
Notice that in the generic case (excluding resonances) a basis of solution is automatically obtained replacing $\tilde k_1$ (and $k_1$) with any of the $\tilde k_j$ (and $k_j$). Indeed, the equation is invariant under permutation of the $j$s.  To be more explicit, we have the following basis of solutions: 
\beq
\left(f_1(z),f_2(z),f_3(z),f_4(z)\right), \label{basis}
\eeq
where $f_j (z)$, with $j>1$, are just obtained from \eqref{Eq_ExatSolution} by replacing $\tilde k_1$ (and $k_1$) with any of the $\tilde k_j$ (and $k_j$). As a consequence, 
we also obtain the general solution of our equation of motion as follows:
\beq\label{solgen}
F(z)=\sum_{i=1}^4 D_i f_i(z),
\eeq
where the constants $D_i$ have to be fixed according to the scattering process one is considering. It is remarkable that the basis is already diagonal in the $\tilde{k}_i$, 
in the sense that the physical modes on the left side (corresponding to $x\to -\infty$, see also the following subsection) are asymptotically represented by just the element of the 
basis with index $j$: $f_j (x) \propto e^{i \tilde k_j x}$ as $x\to -\infty$.\\
 Some physical considerations are mandatory. The aforementioned connection coefficients are responsible of the phenomenon of mode conversion in the scattering process, 
i.e. they show that, from passing from the left, i.e. at $x=-\infty$, with input mode $\tilde k_1$, to the right, i.e. $x=\infty$, with potential output modes $k_j$, $j=1,2,3,4$, 
the S-matrix is not, in general, diagonal, as output modes with $j\not = 1$ are allowed. In making this possible, a fundamental role is played by the Stokes phenomenon, 
which is discussed in the following subsection. It is to be stressed the following point: the Stokes phenomenon is present when at least an irregular singularity appears 
(see e.g. \cite{Haraoka}). In the present case, the equation with $z$ as independent variable displays three Fuchsian singularities, as seen, i.e. three regular singular points 
$z=0,z=1,z=\infty$. Still, by coming back to the original variable $x$, which is the relevant one for the physical problem, one finds that, actually, $x=\pm \infty$ on the real axis correspond to irregular singularities, as essential singularities in $\tanh (\beta x)$ and in $\cosh^{-2} (\beta x)$ appear in the coefficients of the equation itself. This fact is at the root of the Stokes phenomenon in the physical problem at hand.

\subsection{Integral representation and Stokes phenomenon}

By using the integral representation of the Hypergeometric function and changing variable back to $x$, we can write the selected solution of the EOM as
\begin{align}\label{eq_IntegralRepr}
f(x)=&\frac{\Gamma(\beta_1)\Gamma(\beta_2)\Gamma(\beta_3)}{2\pi i \Gamma(\alpha_1) \Gamma(\alpha_2) \Gamma(\alpha_3) \Gamma(\alpha_4)} e^{i \tilde k_1 x}(1+e^{2\tilde \beta x})\times \cr &\int_{\gamma - i\infty}^{\gamma + i \infty} ds \frac{\Gamma(s)\Gamma(\alpha_1-s)\Gamma(\alpha_2-s)\Gamma(\alpha_3-s)\Gamma(\alpha_4-s)}{\Gamma(\beta_1-s)\Gamma(\beta_2-s)\Gamma(\beta_3-s)}(-1)^{-s}e^{-2\tilde \beta x s } \,,
\end{align}
with $0<\gamma<1$. 
The integrand function has simple poles in the $s$-plane that are disposed on five lines parallel to the real axis. The poles are found at
\begin{align*}
s &= \tilde s_{n} := -n \,, \\
s &= s_{1,n} := \alpha_1 + n \,, \\
s &= s_{2,n} := \alpha_2 + n \,, \\
s &= s_{3,n} := \alpha_3 + n \,, \\
s &= s_{4,n} := \alpha_4 + n \,,
\end{align*}
with $n=0,1,2,3...$. The poles are represented in Figure \ref{fig_poles}, where the relative position of the poles is fixed by the following identification of the modes (see Figure \ref{fig_DR}):
\begin{align}\label{e_ModesConvention}
``1"=H \,,\quad ``2"=B \,,\quad ``3"=P \,,\quad ``4"=N \,.
\end{align}
\begin{figure}
	\centering
	\includegraphics[width=.6\textwidth]{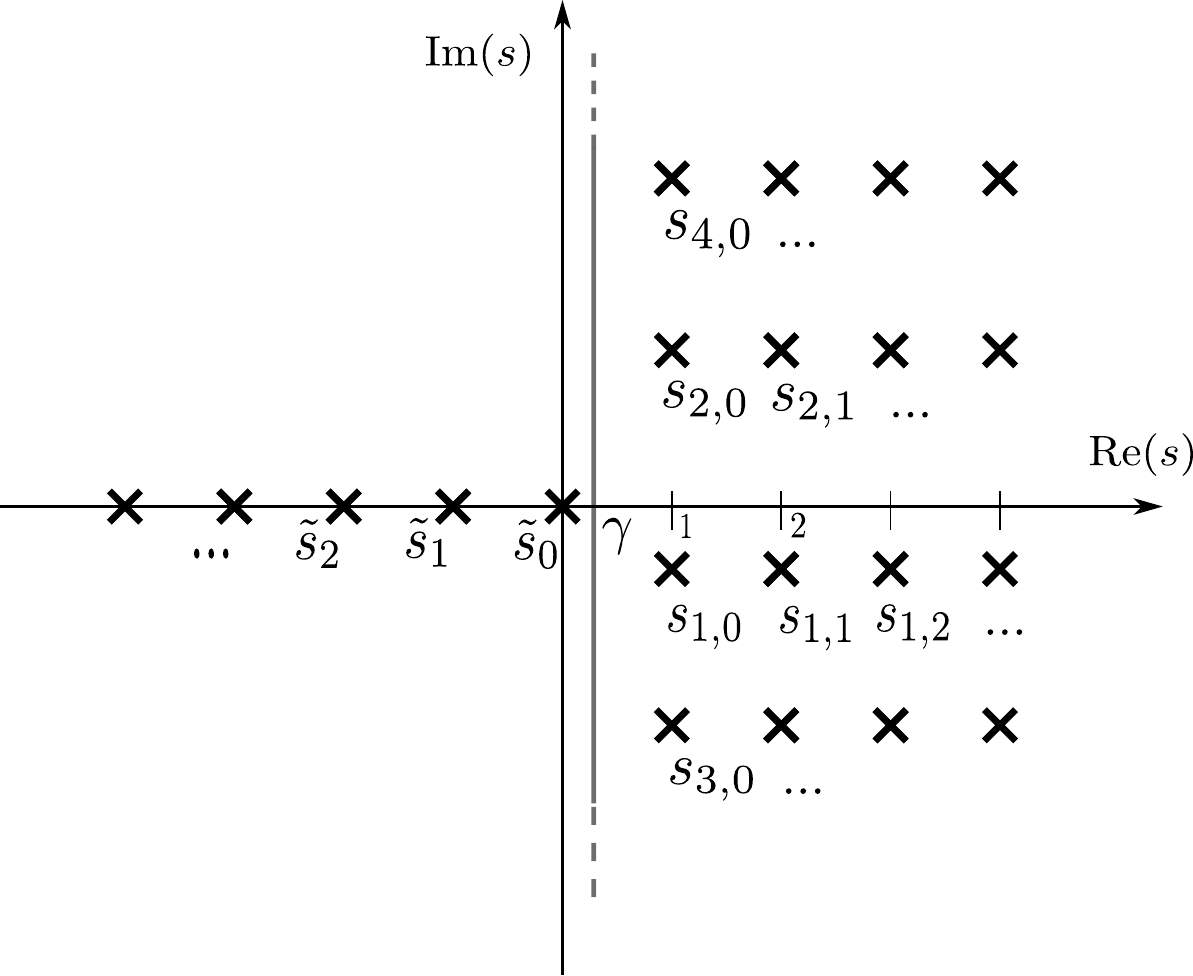}
	\caption{The poles of the integrand function of Eq.(\ref{eq_IntegralRepr}). The grey line represents the path of integration. This figure holds true just when  $k_j$ and $\tilde k_j$ are real for all $j=1,2,3,4$, i.e. in the subcritical case. \label{fig_poles} } 
\end{figure}

We can analytically continue the function $f$ in the complex $x$-plane, in order to study the behaviour for $x\rightarrow \infty e^{i\theta}$ for different angles $\theta$. By writing 
\begin{align*}
x= |x|e^{i\theta}  \,,\quad s=|s|e^{i\phi} \,,
\end{align*}
we see that the integral is convergent only in the half plane defined by
\begin{align*}
\cos\theta \cos \phi - \sin\theta \sin\phi \geq 0 \,,
\end{align*}
which defines a half-plane delimited by the line $\tan\phi = \cot\theta $. Namely, as the angle $\theta$ varies, the corresponding half-plane in the variable $s$ is defined by $\phi\in \left[ \theta - \frac \pi 2, \theta + \frac \pi 2  \right]$.  

Starting from $\theta=\pi$, which represents the solution at $x<0$ (inside the horizon), the integral is defined for $\phi\in\left[\frac \pi 2 , \frac 3 2 \pi\right]$: in this sector, the integral reduces to the sum of the residues at $s=\tilde s_n$.  The sum of the residues gives the series 
\begin{align}\label{Eq_Asympt1}
f(x < 0)=  e^{i \tilde k_1 x}(1+e^{2\tilde \beta x}) \left( 1 + \frac{\alpha_1 \alpha_2 \alpha_3 \alpha_4}{\beta_1 \beta_2 \beta_3} e^{2\tilde \beta x} + \sum_{n\geq2}O((e^{2\tilde \beta x})^n)  \right) \stackrel{x\rightarrow-\infty}{\sim}  e^{i \tilde k_1 x} \,.
\end{align} 
As we move $\theta$, the asymptotic expansion (\ref{Eq_Asympt1}) remains valid until we encounter new poles in the corresponding half-plane in $s$: this happens, as one can see from Figure \ref{fig_contorni}, when $\phi =\arg \alpha_3$ or $\phi = \arg \alpha_4$. As we pass those lines, a new term appears in the asymptotic expansion, corresponding to the residue at the pole $s_{3,0}=\alpha_3$ or $s_{4,0}=\alpha_4$. The appearance of new terms in the asymptotic expansion is known as Stokes Phenomenon: by the previous analysis, thus, we identified a first Stokes sector, given by
\begin{align*}
\theta \in \left[\frac \pi 2  + \arg\alpha_4, \frac 3 2 \pi +  \arg \alpha_3 \right] \,,
\end{align*} 
and the boundaries of this sector are two Stokes lines. 
\begin{figure}
	\centering
	\includegraphics[width=.6\textwidth]{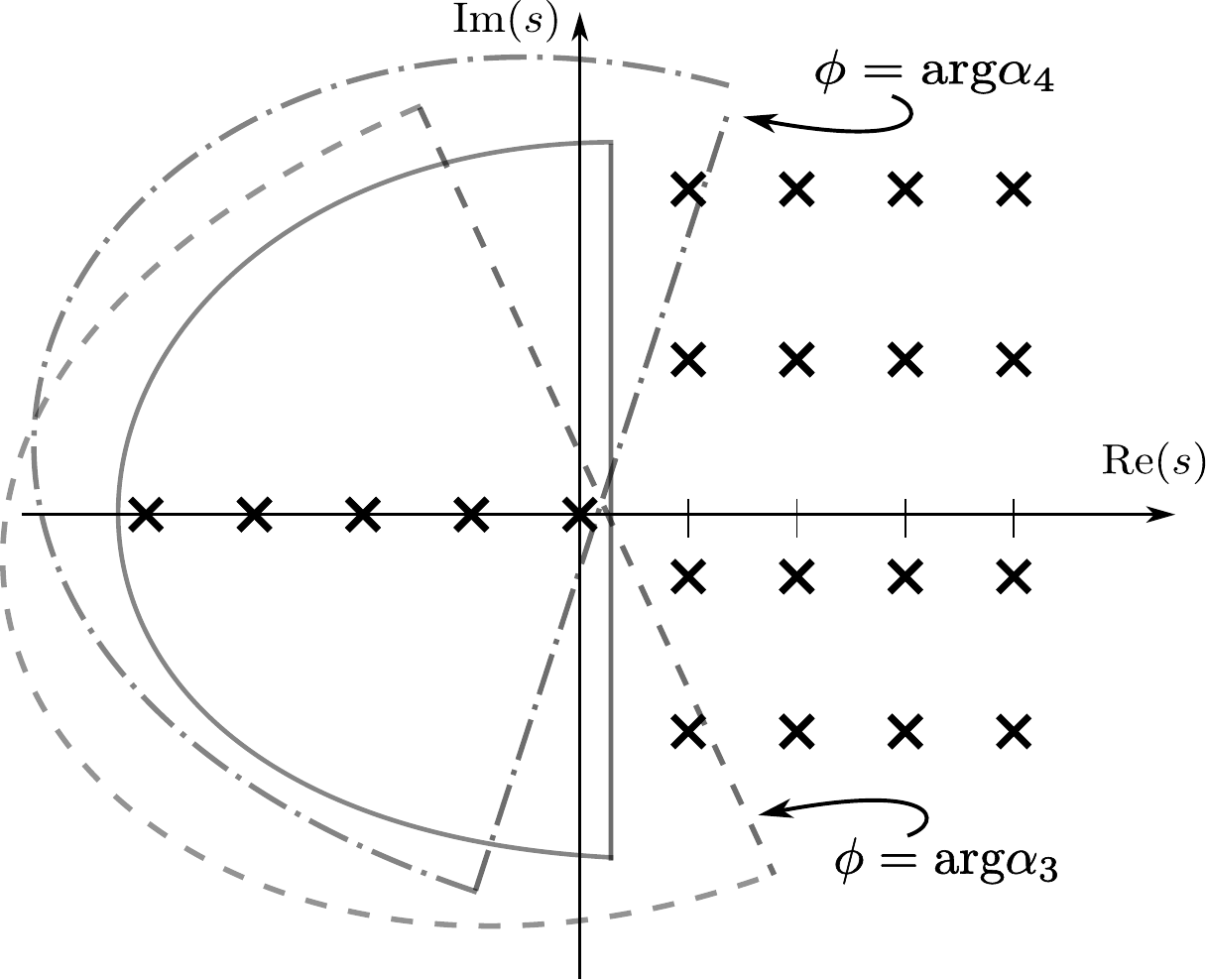}
	\caption{Three contours of integration representing three values of $\theta$. Full line: $\theta=\pi$; dashed line: $\theta= \frac 3 2 \pi +\arg\alpha_3$; dashed-dotted line: $\theta = \frac \pi 2 + \arg\alpha_4$. The latest two are Stokes lines, because the corresponding contour in the $s$-plane (defined by the $\phi$-angle) includes a new pole, giving rise to an additional term in the asymptotic expansion. Also this figure holds true just for the subcritical case.  \label{fig_contorni}}
\end{figure}
As we move $\theta$ past the Stokes line $\theta = \frac 3 2 \pi + \arg\alpha_3$, we include a new pole, $s_{3,0}$, in the contour:  the asymptotic expansion becomes 
\begin{align*}
f(x) = e^{i\tilde k_1 x}\left(1+ O(e^{-2\tilde \beta|x|})\right) + C_{\tilde k_1\rightarrow k_3} e^{ik_3 x} \,,
\end{align*} 
so the new term introduces mode-mixing. This expansion is true until we reach the next pole.\\
Now,  we note that the residues at the poles $s_{j,n}$ are 
$$\text{Res}_{s=s_{j,n}} \sim e^{(i \tilde k_1 + 2\tilde \beta - 2\tilde \beta s_{j,n} ) x } = e^{i k_j x}e^{-nx} \,.$$
Therefore, the contributions of the poles with $n\geq 1$ are negligible as long as we are interested in the asymptotic expansion ($|x|\rightarrow\infty$). From this consideration we understand that the only poles that are related to the Stokes phenomenon are $s=s_{j,0}$. The next Stokes lines are thus met at $\theta = \frac 3 2 \pi + \arg \alpha_1$ or $\theta = \frac \pi 2 + \arg \alpha_2$. 

It is now easy to figure out, by continuing the argument exposed above, all the Stokes lines of the function $f(x)$, which, ordered by increasing $\theta$, correspond to 
\begin{align*}
\theta_1 &= \frac \pi 2 + \arg\alpha_3, \quad
\theta_2 =\frac \pi 2 + \arg\alpha_1, \quad
\theta_3 = \frac \pi 2, \quad
\theta_4 =\frac \pi 2 + \arg\alpha_2, \quad
\theta_5 = \frac \pi 2 + \arg\alpha_4, \\
\theta_6 &=\frac 3 2 \pi + \arg\alpha_3 , \quad
\theta_7 = \frac 3 2 \pi + \arg\alpha_1, \quad
\theta_8 = \frac 3 2 \pi \,\quad
\theta_9 =\frac 3 2 \pi + \arg\alpha_2, \quad
\theta_{10} =\frac 3 2 \pi + \arg\alpha_4 \,.
\end{align*}

The value of $\arg\alpha_j$ depends on the values of the momenta $k_j$ and $\tilde k_1$. The momenta $k_j$, being unperturbed by the background, are always real. On the other hand, as we will discuss also in Section~\ref{sec_transcritical}, $\tilde k_1$ is real in subcritical regime and complex in transcritical regime: in that case we have $\text{Im} \tilde k_1<0$. We can thus evaluate $\arg\alpha_j$ as
\begin{align*}
\arg\alpha_j= \begin{cases*}
\arctan\left(- \frac{k_j - \tilde k_1}{2\tilde \beta}  \right) \,,\quad \text{subcritical regime} \\
\arctan\left(- \frac{k_j - \text{Re}\tilde k_1}{2\tilde \beta- \text{Im}\tilde k_1}  \right) \,,\quad \text{transcritical regime} \,.
\end{cases*}
\end{align*}

\subsection{Subcritical scattering}\label{sec_Subcritical}

The solution (\ref{Eq_ExatSolution}), for $z\rightarrow 0$ (which corresponds to the left asymptotic region $x \rightarrow -\infty$) is
\begin{align}
f \sim z^{i\frac{\tilde k_1}{2\tilde \beta}} = e^{i\tilde k_1 x} \,.
\end{align}
At right infinity $x\rightarrow +\infty$ ($z\rightarrow \infty$) it splits into a sum of plane waves
\begin{align}
f \sim \sum_{j=1}^4 C_{j} z^{i\frac{i k_j}{2\tilde \beta}} = \sum_{j=1}^4 C_{j} e^{i k_j x}\,,
\end{align}
where the connection coefficient $$C_j:= C_{\tilde k_1\rightarrow k_j}$$ can be obtained from (\ref{Eq_ConnCoeff}) by switching $\alpha_1 \leftrightarrow  \alpha_j$. If we put $$\tilde k_1=\tilde k_H\,.$$ the function $f(x)$ represents the scattering of the in-going modes $H$ from left infinity and $P$, $N$ and $B$ from right infinity, which produces an out-going $H$ mode at right infinity. This in the process that, since Hawking's seminal work, is usually considered in black hole physics to deduce the spontaneous particle creation.  Following backwards the out-going $H$ mode, we find that it originates from a mixture of modes: in particular, the coefficient $C_N$ represents the mixing with the negative-norm $N$ mode. As it is shown in \cite{PaperoVariazionale}, the expected number of spontaneously created Hawking particles is
\begin{align}\label{Eq_Monotonic_Nformal}
| N |= \left| \frac{|C_N|^2 v_N \partial_\omega \text{DR}(\omega,k)|_{k_N}}{|C_H|^2 v_H \partial_\omega \text{DR}(\omega,k)|_{k_H}}\right|\,,
\end{align}
Notice that the function $\text{DR}(k)$ defined in (\ref{Eq_CharEq_inf}) is slightly different from the function $\text{DR}(\omega,k)$ that appears in (\ref{Eq_Monotonic_Nformal}), which derives from the normalization of the quantum theory: they differ by a global factor $(\omega + V k)^2$.

\begin{figure}
	\begin{subfigure}{.49\textwidth}
		\centering
		\includegraphics[scale=.45]{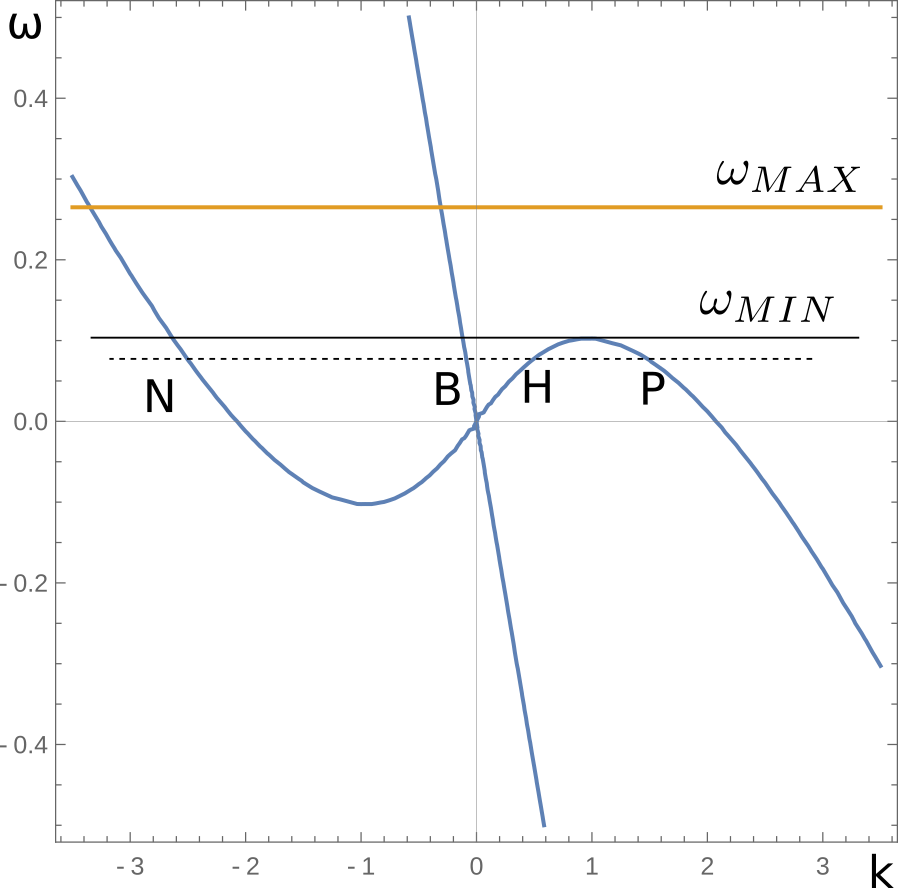}
		\caption{}
		\label{Fig_DRsub}
	\end{subfigure} \hfill
	\begin{subfigure}{.49\textwidth}
		\centering
		\includegraphics[scale=.45]{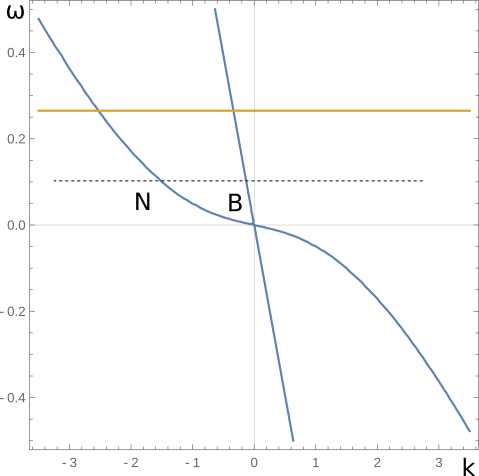}
		\caption{}
		\label{Fig_DRcrit}
	\end{subfigure}
	\caption{\label{Fig_DR2} (A):  The dispersion relation (\ref{Eq_CharEq_zero}) in the perturbed region ($x\rightarrow-\infty$) for $\lambda < \lambda_{\text{crit}}$. The modes $0<\omega<\omega_{MIN}$ do not experience an event horizon (\emph{subcritical regime}). (B): The dispersion relation (\ref{Eq_CharEq_zero}) for $\lambda > \lambda_{\text{crit}}$. In this case, for any $\omega$, the modes $H$ and $P$ become imaginary as they experience an event horizon: this is referred to as \emph{transcritical regime}. }
\end{figure}

The momenta $\tilde k_j(\omega)$ (which correspond to the normal modes at $x \rightarrow -\infty$) depend on the background amplitude $\lambda$. For large enough values of $\lambda$, the momenta $\tilde k_H(\omega)$ and $\tilde k_P(\omega)$ can become imaginary, as can be seen from Figures  \ref{Fig_DRsub} and \ref{Fig_DRcrit}.  In the current literature, the distinction between the subcritcal and transcritical regimes is governed by the presence or absence of a horizon, corresponding to a turning point in the differential equation of motion: the subcritical case is 
when no such turning point appears. On the other hand, the two regimes can be characterized also in a different way: one identifies the transcritical case by the fact that, in the asymptotic dispersion relation, two roots which 
are real in the unperturbed asymptotic region, become complex conjugates in the pertubed region. This is exactly the criterion we adopt in our approach to the 
problem. See also the discussion in subsection \ref{orso-sec}. 

We start considering the case where all $\tilde k_j$ are real (\emph{i.e.} subcritical regime). In this case, the square module of (\ref{Eq_ConnCoeff}) can be written explicitly:
\begin{align}
|C_1|^2=&\frac{(\tilde k_1-\tilde k_2) (\tilde k_1-\tilde k_3) (\tilde k_1-\tilde k_4) (\tilde k_2-k_1) (k_1-\tilde k_3) (k_1-\tilde k_4)
}{(\tilde k_1-k_2) (\tilde k_1-k_3) (\tilde k_1-k_4) (k_1-k_2)
	(k_1-k_3) (k_1-k_4)} \cr 
&\times \frac{\sinh (\frac{\pi\gamma 
	(\tilde k_1-k_2)}{2\beta}) \sinh (\frac{\pi\gamma  (\tilde k_1-k_3)}{2\beta}) \sinh (\frac{\pi\gamma  (\tilde k_1-k_4)}{2\beta})}{\sinh(\frac{\pi \gamma (\tilde k_1-\tilde k_2)}{2\beta}) \sinh(\frac{\pi \gamma  (\tilde k_1-\tilde k_3)}{2\beta}) \sinh(\frac{\pi \gamma  (\tilde k_1-\tilde k_4)}{2\beta}) } \cr  &\times \frac{\sinh (\frac{\pi \gamma  (\tilde k_2-k_1)}{2\beta}) \sinh (\frac{\pi \gamma 
	(\tilde k_3-k_1)}{2\beta}) \sinh (\frac{\pi \gamma  (\tilde k_4-k_1)}{2\beta})}{\sinh(\frac{\pi \gamma (k_1-k_2)}{2\beta}) \sinh(\frac{\pi \gamma  (k_1-k_3)}{2\beta})
	\sinh(\pi  (k_1-k_4))}\,,
\end{align}
and similarly $|C_j|^2$ are obtained by rotations of the indices. 
We want to compare this results with the perturbative expansion we made in \cite{PaperoVariazionale}. We start by writing explicitly 
\begin{align}
|N|= &\frac{(\tilde k_H-k_N) (k_N-\tilde k_P) (k_N-\tilde k_N) (k_N-\tilde k_B) (\omega + k_H V)^2 
}{(\tilde k_H-k_H) (k_H-\tilde k_P) (k_H-\tilde k_N) (k_H-\tilde k_B) (\omega+ k_N V)^2 } \cr
&\times\frac{\sinh \left(\frac{\pi  \gamma (\tilde k_H-k_H)}{2 \beta}\right)
	\sinh \left(\frac{\pi  \gamma
		(\tilde k_P-k_N)}{2 \beta}\right)  \sinh \left(\frac{\pi  \gamma (\tilde k_N-k_N)}{2 \beta}\right)}{\sinh\left(\frac{\pi  \gamma (\tilde k_H-k_N)}{2 \beta}\right) \sinh\left(\frac{\pi  \gamma (\tilde k_P-k_H)}{2 \beta}\right)\sinh\left(\frac{\pi  \gamma (\tilde k_N-k_H)}{2 \beta}\right)} \cr
&\times\frac{\sinh \left(\frac{\pi  \gamma (\tilde k_B-k_N)}{2 \beta}\right) \sinh \left(\frac{\pi  \gamma
		(k_H-k_P)}{2 \beta}\right) \sinh \left(\frac{\pi  \gamma (k_H-k_B)}{2 \beta}\right)}{\sinh\left(\frac{\pi  \gamma (\tilde k_B-k_H)}{2 \beta}\right)\sinh\left(\frac{\pi  \gamma (k_P-k_N)}{2 \beta}\right)
	\sinh\left(\frac{\pi  \gamma (k_N-k_B)}{2 \beta}\right) }\,,
\end{align}
where we have used the expression for the flux factors 
\begin{align}\label{Eq_smartfluxes}
 v_i(\omega)\partial_\omega  \text{DR}|_{k_i}= -\frac{\gamma^2 V^4}{\omega_{lab}^2|_{k_i}}  \prod_{j\neq i} (k_i(\omega) - k_j(\omega))\,,
\end{align} 
that was deduced in the Appendix of \cite{PaperoVariazionale}. We use the low-frequency expressions of the momenta $k_j(\omega)$, written in (\ref{Eq_kH_lowW})-(\ref{Eq_kN_lowW}). Notice that the expressions of $\tilde k_j(\omega)$ are simply obtained by the shift $\mu^2 \mapsto \mu^2 + \lambda$.\\ 
 We stress that, up to this point, our results are {\sl exact}, in the sense that no approximation has been made. Still, in  order to provide analytical 
expressions to the moments $k_j,\tilde{k}_j$, we are forced to introduce some approximations, indeed $k_j,\tilde{k}_j$ are roots of a fourth degree equation: one might 
provide explicit expressions for the corresponding roots, but at the price to write very long and by no means perspicuous expressions. As a consequence, for these roots we use 
approximate expressions for low $\omega$, as discussed in the previous sections.
With the help of Wolfram Mathematica, we compute the leading order of $|N|$ both in $\lambda$ and in $\omega$: we checked that the two limit commute, so the order of the two expansions makes no difference. We obtain
\begin{align}
|N|= \left(\frac{\pi ^2  \lambda ^2 \omega g (g+\mu V) \sinh^2\left(\frac{\pi  \sqrt{g^2-\mu^2 V^2}}{2 \beta  V^2}\right)}{16 \beta ^2 \gamma  \mu (g-\mu V) \left(g^2-\mu^2 V^2\right)^{3/2}} + O(\omega^2)\right) + O(\lambda^3 )\,.
\end{align}
We notice that the qualitative behaviour is the same as in \cite{PaperoVariazionale}, and in particular we have $ N \sim \omega$: this behaviour confirms what was found for the subcritical case also in \cite{CoutantWeinfurtner}. This is a strong confirmation that such a behaviour should be expected in subcritical systems, and it seem not to depend neither on particular approximations, nor on the characteristics of the background function. 

An even more interesting comparison is the estimation of the ``effective temperature'' that the authors found in \cite{PaperoVariazionale} for the subcritical case. The ratio $\frac{|P|}{|N|}$, in the case $\tilde k_j \in \mathbb R$ as before, becomes
\begin{align}
\frac{|P|}{|N|}= &\frac{(\tilde k_H-k_P) (k_P-\tilde k_P) (k_P-\tilde k_N) (k_P-\tilde k_B)  (k_N V+\Omega)^2 }{(\tilde k_H-k_N) (\tilde k_P-k_N) (\tilde k_N-k_N) (\tilde k_B-k_N)
	(k_P V+\Omega)^2} \cr 
&\times \frac{\sinh (\frac{\pi \gamma  (\tilde k_H-k_N)}{2\beta}) \sinh ( \frac{\pi \gamma  (\tilde k_N-k_P)}{2\beta})\sinh
	(\frac{\pi \gamma  (\tilde k_P-k_P)}{2\beta})}{ \sinh(\frac{\pi \gamma (\tilde k_H-k_P)}{2\beta})\sinh(\frac{\pi \gamma (\tilde k_P-k_N)}{2\beta})  \sinh(\frac{\pi \gamma  (\tilde k_N-k_N)}{2\beta})} \cr 
&\times \frac{ \sinh (\frac{\pi \gamma 
	(\tilde k_B-k_P)}{2\beta}) \sinh (\frac{\pi \gamma  (k_H-k_N)}{2\beta}) \sinh (\frac{\pi\gamma  (k_N-k_B)}{2\beta})}{\sinh(\frac{\pi \gamma  (\tilde k_B-k_N)}{2\beta}) \sinh(\frac{\pi\gamma  (k_H-k_P)}{2\beta})  \sinh(\frac{\pi\gamma 
	(k_P-k_B)}{2\beta})}\,.
\end{align}
To estimate the temperature we compute the leading order in $\omega$ of $\log\left(  \frac{|P|}{|N|}\right)$. We use the low-$\omega$ expressions of the modes $k_j$ that we have written in Eq. (\ref{Eq_kH_lowW})-(\ref{Eq_kN_lowW}). The momenta $\tilde k_j$ can be obtained by the switch $\mu^2\rightarrow\mu^2+\lambda$: they can be also written as follows:
{\small
	\begin{align}\label{Eq_kHs_lowW}
	\tilde k_N &= \frac{\sqrt{\mu^2+\lambda} - g V}{g - \sqrt{\mu^2+\lambda} V}\,\omega + O(\omega^3), \\
	\tilde k_B &= -\frac{\sqrt{\mu^2+\lambda} + g V}{g + \sqrt{\mu^2+\lambda} V}\,\omega + O(\omega^3), \\
	\tilde k_P &=  \frac{\sqrt{\lambda_{\text{crit}}-\lambda}}{
		\gamma V } - \left(\frac 1 V + \frac {g^2} {\gamma^2 V^3 (\lambda_{\text{crit}}-\lambda)}\right) \,\omega  - \frac{g^2 (2 g^2 + (\mu^2+\lambda) V^2) }{2\gamma  V^5(\lambda_{\text{crit}}-\lambda)^{5/2} }\,\omega^2 + O(\omega^3) , \\
	\tilde k_H &=  -\frac{\sqrt{\lambda_{\text{crit}}-\lambda}}{
		\gamma V } - \left(\frac 1 V + \frac {g^2} {\gamma^2 V^3 (\lambda_{\text{crit}}-\lambda)}\right) \,\omega  + \frac{g^2 (2 g^2 + (\mu^2 + \lambda ) V^2) }{2\gamma V^5  (\lambda_{\text{crit}}-\lambda)^{5/2} }\,\omega^2 + O(\omega^3)\,,\label{Eq_kNs_lowW}
	\end{align}}
where 
\begin{align}\label{Eq_LambdaC}
\lambda_{\text{crit}}= \frac{g^2- \mu^2 V^2}{V^2} \,.
\end{align}
These expressions make clear that the subcritical  regime (\textit{i.e.} $\tilde k_j\in \mathbb{R}$) corresponds to $\lambda<\lambda_{\text{crit}}$. 

The leading order of $\log\left(  \frac{|P|}{|N|}\right)$ is
{\small
	\begin{align}\label{Eq_LogPoN_Exact}
	&\log\left(  \frac{|P|}{|N|}\right) = \frac{\pi  g  \omega  }{\beta \gamma V \left(g^2-\mu
		^2 V^2\right)^{3/2} \left(g^2-V^2 \left(\lambda +\mu ^2\right)\right)}\Bigg[ g \lambda  V^2 \sqrt{g^2-\mu ^2 V^2}\times \cr
	& \left(\coth \left(\frac{\pi  \left(\sqrt{g^2-\mu ^2 V^2}-\sqrt{g^2-V^2 \left(\lambda +\mu
			^2\right)}\right)}{2 \beta  V^2}\right)+\coth \left(\frac{\pi  \left(\sqrt{g^2-V^2 \left(\lambda +\mu ^2\right)}+\sqrt{g^2-\mu ^2 V^2}\right)}{2 \beta  V^2}\right)\right) \cr 
	&+2
	\left(g^2 V \sqrt{\left(\lambda +\mu ^2\right) \left(g^2-\mu ^2 V^2\right)}-2 g V^2 \left(\lambda +\mu ^2\right) \sqrt{g^2-\mu ^2 V^2}\right) \cr 
	&-2\left(\mu ^2 V^3 \sqrt{\left(\lambda +\mu
		^2\right) \left(g^2-\mu ^2 V^2\right)}+2 g^3 \sqrt{g^2-\mu ^2 V^2}\right) \coth \left(\frac{\pi  \sqrt{g^2-\mu ^2 V^2}}{2 \beta  V^2}\right)\Bigg] + O(\omega^2)\,.\cr
	\end{align}}
In order to compare it to the perturbative result, we take the limit $\lambda \rightarrow 0$:
\begin{align}
\frac{2 g  \omega  \left(\pi  (2 g+\mu  V) \sqrt{g^2-\mu ^2 V^2} \coth \left(\frac{\pi  \sqrt{g^2-\mu ^2 V^2}}{2 \beta  V^2}\right)+2 \beta  g V^2\right)}{\beta \gamma  V
	\left(g^2-\mu ^2 V^2\right)^{3/2}}.
\end{align}
For $\beta \sim 0$, which amounts physically to considering small values for the derivative of the dielectric pulse,  we find:
\begin{align}
&\log\left(  \frac{|P|}{|N|}\right)\approx \beta_{pert} \omega \,, \\
&\beta_{pert}   = \frac {2\pi g (2 g + \mu V)  }{ \beta  \gamma V (g^2 - \mu^2 V^2)} = \frac{\pi \gamma}{\beta} \lim\limits_{\omega\rightarrow 0} \frac{2 k_H - k_P - k_N}{\omega}\,.
\end{align}
This results confirms the validity of the prediction made in \cite{PaperoVariazionale} and that the value ot $T_{pert}$ is not strongly dependent on the peculiarities of the background. We notice, however, that the expression (\ref{Eq_LogPoN_Exact}) allows to study how the temperature depends on $\lambda$: in particular, for $\lambda = \lambda_{\text{crit}}$, we have
\begin{align}
T(\lambda=\lambda_{\text{crit}}) = \left.\frac{\omega}{\log(P/N)}\right|_{\lambda=\lambda_{\text{crit}}} = 0. 
\end{align}
The vanishing of the temperature at $\lambda=\lambda_{\text{crit}}$ is very puzzling, in the sense that thermal particle creation may be found both in the subcritical and in 
the transcritical case, whereas a discontinuous behaviour between the two regimes is just suggested by such a result when $\lambda=\lambda_{\text{crit}}$. It has also been stressed that, in the 
above scheme of approximation for low $\omega$, for $\lambda=\lambda_{\text{crit}}$ one finds a degeneracy, at least at the leading order, of $\tilde{k}_P$ with $\tilde{k}_N$, and a 
singular behaviour of the subleading ones. This kind of phenomenon will be investigated in future analysis.

\subsection{Transcritical scattering}\label{sec_transcritical}

We now consider the solution in the transcritical case, that is $\lambda>\lambda_{\text{crit}}$. In this case, the $\tilde k_H$ and $\tilde k_P$ become complex, as it is shown in Figure \ref{Fig_DRcrit}.  This fact is the direct consequence of the presence of an event horizon: these modes cannot propagate to the left-infinity. The low-$\omega$ expressions are found from Eq. (\ref{Eq_kHs_lowW})-(\ref{Eq_kNs_lowW})  for $\lambda>\lambda_{\text{crit}}$: notice, however, that the modes have now the wrong label, since the mode that is labeled $N$ becomes complex, while the $H$-mode is real, in contradiction with the visual interpretazion of Figure $\ref{Fig_DRcrit}$. Thus, for the transcritical regime, we need to rename the modes in the following way:
{\small
	\begin{align}\label{Eq_kHc_lowW}
	\tilde k_N &= \frac{\sqrt{\mu^2+\lambda} - g V}{g - \sqrt{\mu^2+\lambda} V}\,\omega + O(\omega^3), \\
	\tilde k_B &= -\frac{\sqrt{\mu^2+\lambda} + g V}{g + \sqrt{\mu^2+\lambda} V}\,\omega + O(\omega^3), \\
	\tilde k_P &=  i\frac{\sqrt{\lambda-\lambda_{crit}}}{
		\gamma V } - \left(\frac 1 V - \frac {g^2} {\gamma^2 V^3 (\lambda-\lambda_{crit})}\right) \,\omega  +i \frac{g^2 (2 g^2 + (\mu^2+\lambda) V^2) }{2\gamma  V^5(\lambda-\lambda_{crit})^{5/2} }\,\omega^2 + O(\omega^3) , \\
	\tilde k_H &=  -i\frac{\sqrt{\lambda-\lambda_{crit}}}{
		\gamma V } - \left(\frac 1 V - \frac {g^2} {\gamma^2 V^3 (\lambda - \lambda_{crit})}\right) \,\omega  -i \frac{g^2 (2 g^2 + (\mu^2 + \lambda ) V^2) }{2\gamma V^5  (\lambda - \lambda_{crit})^{5/2} }\,\omega^2 + O(\omega^3)\,.\label{Eq_kNc_lowW}
	\end{align}}
Notice that it holds
\begin{align}\label{Eq_PNconjugates}
(\tilde k_H)^*=\tilde k_P \,.
\end{align}
This is not true just in the low-$\omega$ limit, but for all $\omega$. Indeed, the modes $\tilde k$ are the roots of a fourth order polynomial with real coefficients: since the roots $\tilde k_N$ and $\tilde k_B$ are always real, the other two roots must be either real or complex conjugates.   
Another very relevant observation is that, being the basis \eqref{basis} asymptotically diagonal in the $\tilde{k}_i$, as discussed in the previous sections, 
we have also the possibility to get rid of the unwanted complex and exponentially growing mode, say $\tilde{k}_4$ (about the growing mode cf. e.g. the discussion in \cite{cpf}), simply by imposing that the corresponding coefficient $D_4$ is zero. Actually, in our following discussion, we put only $D_1\not =0$. Note also that the connection coefficients, 
being connecting $\tilde{k}_1$ to $k_i$, $i=1,2,3,4$, cannot resume the aforementioned growing mode.

Thanks to (\ref{Eq_PNconjugates}), we can simplify the expression of $\frac{|P|}{|N|}$: indeed, when computing $\frac{C_P C_P^*}{C_N C_N^*}$, one finds a factor
\begin{align*}
\frac{\Gamma \left(1+\frac{i (\tilde k_H-k_P) \gamma}{2 \beta}\right) \Gamma \left(-\frac{i (\tilde k_P-k_N) \gamma}{2 \beta}\right) \Gamma \left(1+\frac{ i (k_P
		- \tilde k_H^*)\gamma}{2 \beta}\right) \Gamma \left(\frac{i  \left(\tilde k_P^*-k_N\right)\gamma}{2 \beta}\right)}
		{\Gamma \left( 1+\frac{i (\tilde k_H-k_N) \gamma} {2
		\beta}\right) \Gamma \left(-\frac{i (\tilde k_P-k_P) \gamma}{2 \beta}\right) \Gamma \left(1+\frac{i (k_N -i  \tilde k_H^*)\gamma}{2 \beta}\right) \Gamma \left(\frac{i \gamma 
		\left(\tilde k_P^*-k_P\right)}{2 \beta}\right)} \,,
\end{align*}
which, using  (\ref{Eq_PNconjugates}) and recalling $\Gamma(1+z)/\Gamma(z)=z$, reduces to
\begin{align*}
\frac{(k_P
	- \tilde k_P)(\tilde k_P^*-k_P)}{(\tilde k_P-k_N)(\tilde k_P^*-k_N)}=\frac{|k_P-\tilde k_P|^2}{|k_N-\tilde k_P|^2}\,.
\end{align*}
The final exact expression is
\begin{align}
&\frac{|P|}{|N|} = \frac{ (k_P-\tilde k_N) (k_P-\tilde k_B)\,|k_P-\tilde k_P|^2 \, (k_N V+\omega)^2 }{ (k_N-\tilde k_N) (k_N-\tilde k_B)\,|k_N-\tilde k_P|^2 \, (k_P V+\omega)^2}\times\quad\qquad\quad\qquad \cr 
&  \times\frac{\sinh \left(\frac{\pi  \gamma (\tilde k_N-k_P)}{2
		\beta}\right)  \sinh \left(\frac{\pi  \gamma (\tilde k_B-k_P)}{2 \beta}\right)
	\sinh \left(\frac{\pi  \gamma
		(k_H-k_N)}{2 \beta}\right)  \sinh \left(\frac{\pi  \gamma (k_N-k_B)}{2 \beta}\right)}{\sinh\left(\frac{\pi  \gamma (\tilde k_N-k_N)}{2 \beta}\right) \sinh\left(\frac{\pi  \gamma (\tilde k_B-k_N)}{2 \beta}\right) \sinh\left(\frac{\pi  \gamma (k_H-k_P)}{2 \beta}\right)  \sinh \left(\frac{\pi  \gamma (k_P-k_B)}{2 \beta}\right) }
\end{align}
As we did in the previous Section, we compute, to the leading order in $\omega$, 
\begin{align}
\log\left(\frac{|P|}{|N|}\right) = \frac{2 \pi\, \omega \, g^2  \coth \left(\frac{\pi  \sqrt{g^2-\mu^2 V^2}}{2 \beta  V^2}\right)}{\beta \gamma V  \left(g^2-\mu^2 V^2\right) (1-\lambda_{\text{crit}}/\lambda)}+O\left(\omega^2\right)\,.
\end{align}
We can now write the Hawking temperature
\begin{align}
T(\lambda) = \frac{\omega}{ \log\left(\frac{|P|}{|N|}\right)} = \frac{\beta \gamma V  \left(g^2-\mu^2 V^2\right) } {2 \pi\, g^2  \coth \left(\frac{\pi  \sqrt{g^2-\mu^2 V^2}}{2 \beta \gamma V^2}\right)} \left(1-\frac{\lambda_{\text{crit}}}{\lambda}\right)\,.
\end{align}
This result confirms what was found for the subcritical case, that is
\begin{align*}
T(\lambda_{\text{crit}})=0\,.
\end{align*}
More interestingly, in the far critical case $\lambda\gg\lambda_{\text{crit}}$, if we consider $\beta\sim 0$ as previously done, we find
\begin{align}
T_H = \frac{\beta \gamma V  \left(g^2-\mu^2 V^2\right) } {2 \pi\, g^2  } \,,
\end{align}
which coincides with the far critical limit which was obtained in \cite{PaperoVariazionale} using the Orr-Sommerfeld approach.

In Figure \ref{Fig_TcExact} we plot the temperature in units of $T_H$, namely
\begin{align}\label{Eq_ToTH}
\frac{T(\lambda)}{T_H} = \frac{\omega}{T_H \log\left(\frac{|P|}{|N|}\right)} \,,
\end{align}
for various values of $\lambda$, both in subcritical and transcritical case. As in the perturbative approach of the previous section, for the plot we choose $g=1$, $\mu=1.2$; we than choose a near critical pulse velocity $V=0.8$ and a low value $\beta=0.02$. We clearly observe what was predicted in \cite{PaperoVariazionale} using a perturbative approach: the effective temperature computed for $\lambda\ll\lambda_{\text{crit}}$ (subcritical regime) is
\begin{align}
T(\lambda\sim 0)\approx \frac{T_H}{3}\,.
\end{align}
Starting from this value, the temperature decreases for increasing $\lambda$ until it reaches zero at $\lambda=\lambda_{\text{crit}}$. For $\lambda>\lambda_{\text{crit}}$ the temperature starts growing again, and $\lambda \gg\lambda_{\text{crit}}$ it stabilizes at the value $T_H$. 

\begin{figure}
	\begin{subfigure}{\textwidth}
		\centering
		\includegraphics[scale=.95]{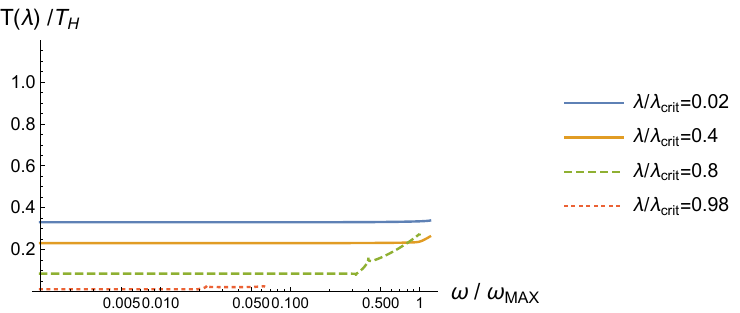}
		\caption{}
		\label{Fig_TsExact}
	\end{subfigure}\\ \vspace{4mm}
	\begin{subfigure}{\textwidth}
		\centering
		\includegraphics[scale=.95]{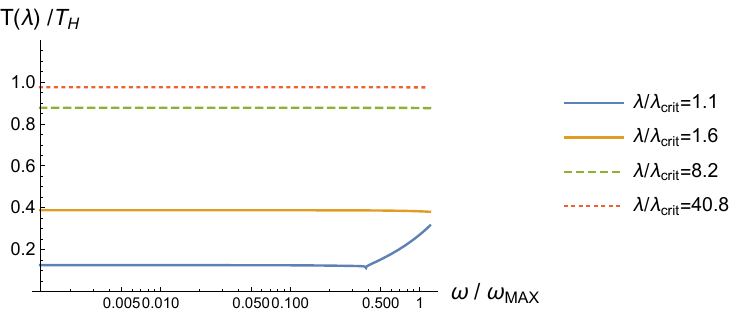}
		\caption{}
		\label{Fig_TcExact}
	\end{subfigure}
	\caption{The temperature $T(\lambda)$ in units of $T_H$, as defined in (\ref{Eq_ToTH}), for different values of $\lambda$, in  subcritical  regime (A)  and   transcritical regime (B). The plots are made for a near critical pulse velocity and a low value of the parameter $\beta$. Starting from $T(\lambda\sim 0)\approx T_H/3$, the temperature decreases till reaching zero for $\lambda=\lambda_{\text{crit}}$; for $\lambda>0$ (transcritical) the temperature increases again, reaching $T_H$ for $\lambda\gg\lambda_{\text{crit}}$.   }
	\label{Fig_TExact}
\end{figure}

So far, the exact solution we provided has confirmed the predictions that were made using different approximations in different regimes. In the future, an even deeper study of this solution may allow to describe precisely the transition between the subcritical and the transcritical regime, the onset of termality and formation of the event horizon.

\section{The original $\phi\psi$-model, exact solutions and thermality}
\label{original-fipsi}

In this section, we provide results concerning a reduction of the so-called Hopfield model, which represents an effective description 
of the interaction between the electromagnetic field and a dielectric medium. In particular, atoms and molecules of the dielectric 
are replaced by a mesoscopic polarization field, still providing an efficient physical description.  
The electromagnetic Lagrangian for the full Hopfield model is quite involved, and has been discussed, by means of different theoretical 
tools, in \cite{hopfield-covariant,hopfield-kerr}. A simplified model, introduced in \cite{hopfield-hawking},  can be related to the two dimensional reduction of the Hopfield model adopted in \cite{finazzi-carusotto-pra13}, and is such that the electromagnetic field and the polarization field are simulated by a pair of scalar fields, $\varphi$, and $\psi$ respectively, in the so-called $\varphi \psi$--model. Despite its simplification, it is still set-up in such a way that we get exactly the same dispersion relation and, moreover, we can simulate the same coupling as in the full case.  
Its Lagrangian is 
\begin{align}
{\mathcal L}_{\varphi\psi} = \frac{1}{2} (\partial_\mu \varphi)(\partial^\mu \varphi)+\frac{1}{2\chi \omega_0^2} \left[ (v^\alpha \partial_\alpha \psi)^2 - \omega_0^2 \psi^2 \right] -\frac{g}{c} (v^\alpha \partial_\alpha \psi) \varphi +\frac{\lambda}{4!} \psi^4,
\label{Lagrangian}
\end{align}
where $\chi$ plays the role of the dielectric susceptibility, $v^\mu$ is the usual four-velocity vector of the dielectric, $\omega_0$ is the proper frequency  of the medium, and $g$ is the coupling constant between the fields. The latter constant is henceforth put equal to one, as its original motivation (see \cite{hopfield-hawking}) can be relaxed without problems in a more advanced discussion (cf. also \cite{master}). As shown in \cite{belcaccia-universe}, we may introduce the above fourth-order non-linear term in the polarization field $\psi$ in the Lagrangian. 
Herein, we assume $\lambda>0$.\\
By extending our analysis, and on the grounds of the previous sections, we adopt a phenomenological model where we can leave room for a spacetime dependence of the microscopic parameters $\chi,\omega_0$, in such a way that $\chi \omega_0^2 $ is a constant. 
The equations of motion are
\begin{align}
\square \varphi + \frac{1}{c}  (v^\alpha \partial_\alpha \psi)=&0,\\
\frac{1}{\chi \omega_0^2} (v^\alpha \partial_\alpha)^2 \psi + \frac{1}{\chi} \psi - \frac{1}{c} v^\alpha \partial_\alpha  \varphi=&0. 
\end{align}
In particular, we define 
\beq
\epsilon^2:=\frac{1}{\chi \omega_0^2},
\eeq
which corresponds to the parameter appearing in the Orr-Sommerfeld like equation (master equation, cf. \cite{master}). 
We can separate the above system, obtaining equations involving only one of the fields $\varphi, \psi$. We can also separate the equations 
for $\varphi, \psi$, and, in order to maintain the same line of reasoning as in the previous sections, we 
focus on $\psi$, obtaining 
\beq
\epsilon^2 \square (v^\alpha \partial_\alpha)^2 \psi+\square \frac{1}{\chi} \psi+\frac{1}{c^2} (v^\alpha \partial_\alpha)^2 \psi=0. 
\label{eqpsi}
\eeq
Let us also choose, as usual, the comoving frame; If we put $\psi (t,x)=e^{-i \omega t} h(x)$, we obtain 
\begin{align}
\label{eq-for-h}
&\epsilon^2 h^{(4)}(x)+2 i \epsilon^2 \frac{\omega}{v}   h^{(3)}(x)+\left(\frac{1- \gamma^2 \frac{v^2}{c^2} \chi(x)}{\gamma^2 v^2 \chi(x)}-
\epsilon^2 \frac{w^2}{\gamma^2 v^2}\right)h^{(2)}(x)+
\left(-\frac{2 i \omega}{c^2 v}-\frac{2 \chi'(x)}{\gamma^2 v^2 \chi^2 (x)}+\frac{2 \epsilon^2 i \omega^3}{c^2 v}\right) h^{(1)}(x)
+\cr
&\left( \frac{\omega^2}{c^2 \gamma^2 v^2 \chi (x)}(1+\gamma^2 \chi (x)) +\frac{2 \chi'^2(x)-\chi(x) \chi''(x)}{\gamma^2 v^2 \chi^3 (x)}
-\epsilon^2 \frac{\omega^4}{v^2 c^2} \right) h(x)=0,
\end{align}
where $\chi'(x),\chi''(x)$ indicate the first and the second derivative wrt $x$.
We stress that, with respect to \cite{master}, we do not eliminate the third order term, as we don't need to grant a Orr-Sommerfeld form for 
our equation of motion, as we are going to compute exact solutions, i.e., solutions which do not depend on the smallness of the parameter $\epsilon$. 
In the following, we assume the monotone profile in the comoving frame
\beq
\frac{1}{\chi(x)}=\frac{1}{\chi_0}-\frac{1}{2} \lambda (1-\tanh (\tilde \beta x)),
\label{profile-os}
\eeq
where $\chi_0$ is a constant value of the dielectric susceptibility and we define the parameter $\tilde \beta=\frac{\beta}{\gamma}$ as in the previous model. We discuss some physical consequence of our choice in Appendix \ref{refractive-sec}.

\subsection{A Fuchsian framework for the equation of motion}

We consider the following change of variable: $z=-e^{2\tilde  \beta x}$. As a consequence, we obtain 
\beq
\label{unchiz}
\frac{1}{\chi(z)}=\frac{1}{\chi_0}-\frac{\lambda}{1-z},
\eeq
and
\begin{align}
\label{equazeta} 
&16 \epsilon^2 \tilde \beta^4 z^4 h^{(4)}(z)+16 \epsilon^2 \frac{(6 \tilde \beta v+i \omega)}{v} \tilde \beta^3 z^3 h^{(3)}(z)+\cr
&4 
\left(\frac{-1+\epsilon^2 \omega^2}{c^2}+ \frac{1-z+\chi_0\left( -\lambda+\epsilon^2 \gamma^2 (28 \tilde \beta^2 v^2+12 i \tilde \beta v \omega-\omega^2) \right) (1+z)}{\chi_0 \gamma^2 v^2 (1-z)}\right) \tilde \beta^2 z^2 h^{(2)}(z)+\cr
& 4 \left( 4 \tilde \beta^3 \epsilon^2+\frac{4 i \tilde \beta^2 \epsilon^2 \omega}{v} +\frac{\tilde \beta(-1+\epsilon^2 \omega^2)}{c^2}+\frac{i \omega(-1+\epsilon^2 \omega^2)}{c^2 v}
+\frac{\tilde \beta}{\gamma^2 v^2} \left(\frac{1}{\chi_0}-\frac{\lambda (1+z)}{(1-z)^2}-\epsilon^2 \gamma^2 \omega^2 \right) \right) \tilde \beta z h^{(1)}(z)+\cr
&\left( -\frac{4 \tilde \beta^2 \lambda z (z+1)}{\gamma^2 v^2 (1-z)^3}-\frac{\epsilon^2 \omega^4}{v^2 c^2} +\frac{\omega^2 (1-z-\chi_0 \lambda)}{c^2 \chi_0 \gamma^2 v^2 (1-z)}
+\frac{\omega^2}{c^2 v^2}\right) h(z)=0.
\end{align}

We start looking for local solutions, along the path sketched for \eqref{Eq_Monotonic_FuchsianEQ}, around $z=\infty$. By introducing $t=1/z$, a series expansion for the solutions can be 
provided in the following form: 
\begin{align}\label{Eq_FormalSolution_inf_psi}
f(t) = t^{-i \alpha} \sum_{n=0}^{\infty} c_n t^n \,.
\end{align}
The characteristic equation for the exponent $k:=2\tilde \beta \alpha$ is
\begin{align}\label{Eq_CharEq_inf_psi}
\text{DR}(k)&:= 
(\omega^2 - \chi_0 \gamma^2 (k v + \omega)^2 (-1 + \epsilon^2 \omega^2) + 
c^2 k^2 (-1 + \chi_0 \epsilon^2 \gamma^2 (k v + \omega)^2))=0.
\end{align}
As in the previous sections, except possibly for a zero measure set in the space of available parameters 
in our model, we have four distinct roots which, moreover, not differ each other by an integer value. As 
a consequence, the spectral type is $(1111)$.\\

At $z=0$, we get four independent local solutions of the form
\begin{align}\label{Eq_FormalSolution_zero_psi}
f(z) = z^{i \tilde\alpha} \sum_{n=0}^{\infty} c_n z^n \,,
\end{align}
where $\tilde k :=2\tilde \beta \tilde\alpha$ satisfies 
\begin{align}\label{Eq_CharEq_zero_psi}
\text{DR}_0(\tilde k) & :=(\omega^2 - \chi_0 \lambda \omega^2 - \chi_0 \gamma^2 (\tilde k v + \omega)^2 (-1 + \epsilon^2 \omega^2) + 
c^2 {\tilde k}^2 (-1 + 
\chi_0 (\lambda + \epsilon^2 \gamma^2 (\tilde k v + \omega)^2)))=0,
\end{align}
The spectral type at $z=0$ is again $(1111)$, again almost everywhere in the space of available parameters appearing in our model.

Also in this case, at $z=1$ the so-called resonant case\cite{Haraoka} is verified. Let us define $y:=z-1$. Then, the characteristic equation for solutions of the form
\begin{align}\label{Eq_FormalSolution_1_psi}
f(y) = y^{a} \sum_{n=0}^{\infty} c_n y^n 
\end{align}
has four integer solutions $a=0,1,2,3$. 
 This situation, again, requires 
	a particular study, that we perform by means of the Frobenius method.  
	
	We can show that there exist three independent integer solutions

		\begin{align}\label{Eq_localsol_1_psi}
		u_1(y) &= y^3+y^4 \left[ -\frac{3}{2}-\frac{\lambda}{48 \tilde \beta^2 \epsilon^2 \gamma^2 v^2}-i\frac{\omega}{4\tilde \beta v}\right]  + o(y^4)\\
		u_2(y) &= y^2 +y^3 \left[ -2-\frac{\lambda}{24 \tilde \beta^2 \epsilon^2 \gamma^2 v^2}-i\frac{\omega}{3\tilde \beta v}\right] + o(y^3) \\
		u_3(y) &= y\; +y^2 \left[-3-\frac{\lambda}{8 \tilde \beta^2 \epsilon^2 \gamma^2 v^2}-i\frac{\omega}{2 \tilde \beta v}\right] + o(y^2) \,,
		\end{align}
		and one logarithmic solution 
		\begin{align}\label{Eq_localsol_4_psi}
		&u_0(y) =1+y \left[-6+\frac{\lambda}{4 \tilde \beta^2 \epsilon^2 \gamma^2 v^2}-i\frac{\omega}{\tilde \beta v}\right]+o(y)\cr 
		& +\log (y) \Big(R_1 u_1(y)+R_2 u_2(y)+R_3 u_3(y)\Big) \,,
		\end{align} 
		
		where 
		{\small 
			\begin{align}
			R_3 &= -\frac{\lambda}{4 \tilde \beta^2 \epsilon^2 \gamma^2 v^2} \\
			R_2 &= -\frac{5\lambda}{16\tilde \beta^2 \epsilon^2 \gamma^2 v^2} \\
			R_1 &= -\frac{\lambda \left(36\tilde \beta^2 c^2+ \omega^2\right)}{288\tilde \beta^4 c^2 \epsilon^2 \gamma^2 v^2} \,.
			\end{align}}

	\subsection{Monodromy and rigidity}

	The monodromy matrix of the solutions $(u_0(y),u_1(y),u_2(y),u_3(y))$ at $z=1$ is easily computed as
	\begin{align}
	M_1 = \left(
	\begin{array}{cccc}
	1 & 2 \pi i  R_1 & 2  \pi i R_2 & 2 \pi i R_3 \\
	0 & 1 & 0 & 0 \\
	0 & 0 & 1 & 0 \\
	0 & 0 & 0 & 1 \\
	\end{array}
	\right) \,,
	\end{align}
	whose Jordan form is 
	\begin{align}
	J_{M_1} = \left(
	\begin{array}{cccc}
	1 & 0 & 0 & 0 \\
	0 & 1 & 0 & 0 \\
	0 & 0 & 1 & 1 \\
	0 & 0 & 0 & 1 \\
	\end{array}
	\right)\,.
	\end{align}
	The monodromy at $z=0$ and $z=\infty$ are, also in this case, represented respectively by the diagonal matrices
	\begin{align}
	M_0 &= \text{diag}\left(e^{i2\pi \tilde \alpha_1},e^{i2\pi \tilde \alpha_2},e^{i2\pi \tilde \alpha_3},e^{i2\pi \tilde \alpha_4}\right)\,,\\
	M_\infty &= \text{diag}\left(e^{-i2\pi  \alpha_1},e^{-i2\pi  \alpha_2},e^{-i2\pi \alpha_3},e^{-i2\pi \alpha_4}\right)\,,
	\end{align}
	whose Jordan form is $J_{M_0}=J_{M_\infty}= \mathbb I_4$.
	As a consequence, the spectral type of the equation at $z=1$ is $(3,1)$, and  
	the spectral type of the equation is $((1111),(31),(1111))$, i.e. it is rigid as in the case discussed in the previous section.

\subsection{Gauge Transformation and Riemann scheme}

As in the previous sections for the simplified model, we can obtain a solution of the form
\begin{align}
f(z):= z^{i \frac {\tilde  k_1}{2 \tilde \beta}}(z-1)\, u(z)\,,
\end{align}
where $\tilde k_1$ satisfies $$\text{DR}_0(\tilde k_1)=0\,.$$  The function $u(z)$ now satisfies an equation which is analog to 
\eqref{Eq_FuchsianEQ_gauge2}, which we avoid to write explicitly, as it is quite long. What happens is that one may verify that for the 
solution $u(z)$, also in the case of the standard $\phi\psi$-model, the same equation \eqref{Eq_FuchsianEQ_DR} holds true, where now the 
dispersion relations are the ones in \eqref{Eq_CharEq_inf_psi} and in \eqref{Eq_CharEq_zero_psi} respectively. Also the so-called Riemann Scheme of the equation 
is the same as in \eqref{Eq_Monotonic_RiemannScheme}. 
Also in this case, by letting 
\begin{align}
\alpha_i &:= 1-\frac i {2\tilde \beta}( k_i - \tilde k_1) \,,\quad i=1,2,3,4\\ 
\beta_j &:= 1 - \frac i {2\tilde \beta}(\tilde k_{j+1} - \tilde k_1)\,,\quad j=1,2,3\,,
\end{align}
we can write Eq. (\ref{Eq_Monotonic_RiemannScheme}) as
\begin{align}
\left[\begin{matrix}
z=0 & z=1 & z=\infty \\
0	& 0	& \alpha_1 \\
1-\beta_1	& 1	& \alpha_2 \\
1-\beta_2	& 2	& \alpha_3 \\
1-\beta_3 & -\beta_4	& \alpha_4
\end{matrix}\right]\,.
\end{align}
which, as discussed in the previous sections, corresponds to  the Riemann scheme of the hypergeometric function $$_4F_3(\alpha_1,\alpha_2,\alpha_3,\alpha_4;\beta_1,\beta_2,\beta_3;z)$$ in the standard form.

\subsection{The Hawking temperature}
The scattering coefficients $C_j$ for the Hopfield model have the same form as the ones of the previous Sections, a part for the different expressions of the momenta $k_j(\omega)$ and $\tilde k_j(\omega) $. As in the previous section, one is able to deal with analytical expressions for the momenta only in suitable limits, and this is the strategy that is adopted also in the present case. 
 
In the low-$\omega$ limit, the right-infinity modes are
\begin{align}
k_H &= \frac{ \chi_0 \gamma^2 v+\sqrt{c^2 (1+\chi_0)}}{c^2-\chi_0
	\gamma^2 v^2}\,\omega  + o(\omega)\\
k_B &= \frac{\chi_0 \gamma^2 v-\sqrt{c^2(1+\chi_0)}}{c^2-\chi_0
	\gamma^2 v^2}\,\omega + o(\omega) \\
k_P&= \frac{\sqrt{c^2-\chi_0 \gamma^2 v^2}}{c \sqrt{\chi_0} {\epsilon} \gamma v}+\frac{c^2 \omega}{v \left(\chi_0 \gamma^2
	v^2-c^2\right)} + o(\omega) \\
k_N&=-\frac{\sqrt{c^2-\chi_0 \gamma^2 v^2}}{c \sqrt{\chi_0} {\epsilon}
	\gamma v} + \frac{c^2 \omega}{v \left(\chi_0 \gamma^2 v^2-c^2\right)} + o(\omega) \,,
\end{align}
while the left-infinity modes are given by
\begin{align}
\tilde k_H &= \frac{ \chi_0 \gamma^2 v + \sqrt{ c^2 \left(1 + \chi_0 + \chi_0^2 (\lambda-1) \lambda-2 \chi_0 \lambda\right)}}{c^2-\chi_0 \gamma^2 v^2 - c^2 \chi_0
	\lambda}\,\omega + o(\omega)\\
\tilde k_B &= \frac{ \chi_0 \gamma^2 v - \sqrt{ c^2 \left(1 + \chi_0 + \chi_0^2 (\lambda-1) \lambda-2 \chi_0 \lambda\right)}}{c^2-\chi_0 \gamma^2 v^2 - c^2 \chi_0
	\lambda}\,\omega + o(\omega) \\
\tilde k_P &= +\frac{ \sqrt{\chi_0(\lambda_{\text{crit}} -\lambda)}}{  {\epsilon} \gamma v} - \frac{ \left(1- \chi_0 \lambda\right)\,\omega}{v  \chi_0\left(\lambda_{\text{crit}}-\lambda\right)}\\
\tilde k_N &= -\frac{ \sqrt{\chi_0(\lambda_{\text{crit}} -\lambda)}}{  {\epsilon} \gamma v} - \frac{ \left(1- \chi_0 \lambda\right)\,\omega}{v  \chi_0\left(\lambda_{\text{crit}}-\lambda\right)}  \\
\end{align}
The transition between subcritical and transcritical regime happens at $\lambda=\lambda_{\text{crit}}$, which in this case is
\begin{align}
\lambda_{\text{crit}}= \frac{1-\chi_0 \gamma^2 \frac{v^2}{c^2}}{\chi_0}\,.
\end{align}
 We assume the condition $\lambda_{\text{crit}}>0$.

As we did in Sections \ref{sec_Subcritical} and \ref{sec_transcritical}, we expand the factor $\frac{|P|}{|N|}$ in order to find an expression of the Hawking temperature. The flux factors that appear in the expressions of $|P|$ and $|N|$ (see Eq. (\ref{Eq_Monotonic_Nformal})) for the Hopfield model are given by
\begin{align}\label{Eq_smartfluxes2}
v_i(\omega)\partial_\omega  \text{DR}|_{k_i}= -\frac{\gamma^2 V^4}{\omega^2- k_i(\omega)^2}  \prod_{j\neq i} (k_i(\omega) - k_j(\omega))\,.
\end{align} 

\subsubsection*{Subcritical case}
In the subcritical case, expanding for low-$\omega$ and low-$\lambda$ we find
\begin{align}
\log\left(\frac{|P|}{|N|}\right) = \left(\beta_{sub} +o(\lambda)\right)\omega + o(\omega)\,,
\end{align}
where 
\begin{align}\label{eq_BetaSub}
\beta_{sub} = \frac{2 \left(\pi  \left(1-v^2 (\text{$\chi_0$}+1)\right) \left(\gamma ^2 v^2 \text{$\chi_0 $}+v \sqrt{\text{$\chi_0 $}+1}+1\right) \coth
	\left(\frac{\pi  \gamma  \sqrt{1-v^2 (\text{$\chi_0 $}+1)}}{2 \beta  v \sqrt{\text{$\chi_0 $}} \epsilon }\right)+2 \beta  \gamma  v^3
	\text{$\chi_0 $}^{3/2} \epsilon  \sqrt{1-v^2 (\text{$\chi_0 $}+1)}\right)}{\beta  \gamma  v \left((\chi_0+1) v^2-1\right)^2}
\end{align}
Once again, for $\beta \sim 0$ we find 
\begin{align}
\beta_{sub}\approx \frac{2 \pi  \left(\gamma ^2 v^2 \text{$\chi_0 $}+v \sqrt{\text{$\chi_0 $}+1}+1\right)}{\beta  \gamma  v \left(1-v^2 (\text{$\chi_0 $}+1)\right)}\,,
\end{align}
which is the same result that was found in \cite{PaperoVariazionale} with a perturbative approach. Notice that, for this model, taking the limit $\beta\sim 0$ in (\ref{eq_BetaSub}) is the same as taking the limit $\epsilon\sim 0$: this corresponds to the weak dispersion limit, a situation that is often studied in literature.

We can expand $\log{\frac{|P|}{|N|}}$ for $\lambda \lesssim\lambda_{\text{crit}}$, we find
\begin{align}
\log{\frac{|P|}{|N|}} = \left(\frac{A}{\frac{\lambda_{crit}}{\lambda} - 1} + B + O\left(\frac{\lambda_{crit}}{\lambda}-1\right) \right)\omega + o(\omega)\,,
\end{align}
with determined factors $A$ e $B$. From this expansion we deduce the temperature in the near critical regime is
\begin{align}
T_{nc}(\lambda) \approx \frac{\beta \left(1-v^2 (\text{$\chi_0 $}+1)\right) \tanh \left(\frac{\pi  \gamma  \sqrt{1-v^2 (\text{$\chi_0 $}+1)}}{2 \beta  v
		\sqrt{\text{$\chi_0 $}} \epsilon }\right)}{4 \pi  \gamma  v \text{$\chi_0$}}\left(\frac{\lambda_{crit}}{\lambda}-1\right)\,.
\end{align}
In particular,  notice as before that
$$ T(\lambda=\lambda_{\text{crit}})=0\,. $$
The same considerations as in the previous section hold true about this point.

\subsubsection*{Transcritical case}
The transcritical regime is reached for $\lambda >\lambda_{\text{crit}}$. In this case, the low-$\omega$ expansion gives a much simpler result,
\begin{align}
\log{\frac{|P|}{|N|}} = \beta(\lambda) \omega + o(\omega) \,,
\end{align} 
where
\begin{align}
\beta(\lambda) = \frac{2 \pi  \gamma\,  v \chi_0 \coth \left(\frac{\pi  \gamma  \sqrt{1-v^2 (\chi_0+1)}}{2 \beta  v \sqrt{\chi_0}
		\epsilon }\right)}{\beta   \left(1-v^2 (\chi_0+1)\right)(1-\frac{\lambda_{\text{crit}}}{\lambda} ) }\,.
\end{align}
Thus, the Haking temperature in transcritical regime is
\begin{align}
\label{temp-cr}
T_c(\lambda)=\frac{\beta   \left(1-v^2 (\chi_0+1)\right) \tanh \left(\frac{\pi  \gamma  \sqrt{1-v^2 (\chi_0+1)}}{2 \beta  v
		\sqrt{\chi_0} \epsilon }\right)}{2 \pi  \gamma  v \chi_0}\left( 1- \frac{\lambda_{\text{crit}}}{\lambda}\right)
\end{align}
In the limit $\lambda\gg\lambda_{\text{crit}}$  (see also Appendix \ref{refractive-sec}), we reach the limit temperature
\begin{align}
T_H = 	\frac{\beta   \left(1-v^2 (\chi_0+1)\right) \tanh \left(\frac{\pi  \gamma  \sqrt{1-v^2 (\chi_0+1)}}{2 \beta  v
		\sqrt{\chi_0} \epsilon }\right)}{2 \pi  \gamma  v \chi_0}\,.
\end{align}
This result, for $\beta\sim 0$, again coincides with what was found in \cite{PaperoVariazionale} using the Orr-Sommerfeld approach. 
 It is also to be stressed, as for the subcritical case, that if, in place of $\beta\sim 0$, one considers in the last equation $\epsilon \sim 0$, i.e. the usual weak dispersion limit which 
is commonly adopted in the literature on the dispersive analog Hawking effect, we get the same result. In order to provide a more extensive comparison with the Orr-Sommerfeld 
approach and some more insights, in the 
following subsection we sketch the basic calculations involved.

\subsection{The transcritical case in the Orr-Sommerfeld picture} \label{orso-sec}

The separated equation of motion for the spatial part of the polarization field $h (x)$, has been displayed in \eqref{eq-for-h}, 
where, in the present case, the specific profile \eqref{profile-os} is understood. 
We eliminate, as usual \cite{master}, the third order term by putting $h (x)=\exp(-2 i \frac{\omega}{v} x) f(x)$. Then we obtain the following equation:
\beq
\epsilon^2 f^{(4)} (x)+\sum_{i=0}^2 p_{3-i} (x,\epsilon)  f^{(2-i)} (x)=0,
\eeq
where the coefficients $p_i (x,\epsilon)$ are, in the Orr-Sommerfeld approach, analytic functions in $\epsilon$: 
$$p_i (x,\epsilon)=\sum_{n=0}^\infty \epsilon^n p_{i n} (x).$$ 
A real turning point $x=x_{tp}$, i.e. an horizon, is found when $p_{30}(x_{tp})=0$. See \cite{master} and references therein, with particular focus on the papers by Nishimoto. We get
\beq
p_3 (x,\epsilon)=\left[\left(1-\chi_0 \gamma^2 \frac{v^2}{c^2}\right)-\frac{\lambda}{2 \gamma^2 v^2} (1-\tanh (\tilde \beta x))\right]+\epsilon^2 \omega^2  \frac{1}{2 v^2}(1+2 \frac{v^2}{c^2}).
\eeq
We can easily identify the turning point by solving $p_{30}(x=x_{tp})=0$.We find
\beq
\tilde \beta x_{tp}= \taninv\left(1- 2\frac{\lambda_{\text{crit}}}{\lambda}\right)
\eeq
It interesting to notice that, by assuming $\lambda$ and $\lambda_{\text{crit}}$ both positive, as in our previous analysis, the condition $\lambda > \lambda_{\text{crit}}$ amounts to the 
reality of the critical point at hand, i.e. to $(1- 2\frac{\lambda_{\text{crit}}}{\lambda})\in(-1,1)$. This is in agreement with the assumption of transcritical case and with the interpretation of $\lambda_{\text{crit}}$. 
We also have \cite{master}
\beq
T_H=\frac{\gamma^2 v^2 |n'(x_{tp})|}{2\pi},
\eeq
where the refractive index is given by $n(x)=\sqrt{1+\chi(x)}$. 
Then, after restoring the parameter $\beta =\gamma \tilde \beta$,  we obtain 
\beq
T_H=\frac{\beta c^2 \left(1-(1+\chi_0)  \frac{v^2}{c^2}\right)\left( 1- \frac{\lambda_{\text{crit}}}{\lambda}\right)}{2 \pi v \gamma \chi_0}, 
\eeq
which is in perfect agreement with \eqref{temp-cr} if one considers the limit as $\epsilon \to 0$ in \eqref{temp-cr}, because, trivially, the 
contribution of the factor involving the hyperbolic tangent goes to 1 in that limit.

\section{Conclusions}\label{conclusions}

In the framework of the analog Hawking effect in dielectric media, we have taken into account both the Cauchy model, which has the characteristic to be as simple as possible, and the original $\phi\psi$-model, with the explicit aim to find out exact solutions for the scattering problem for a suitable but physically meaningful monotone profile for the dielectric refractive index perturbation. On the one hand, this has required to embed the physical problem, from a mathematical point of view,  in the framework of Fuchsian equations on the Riemann sphere. We have first 
introduced the complex variable $z$, and obtained a fourth order equation displaying three regular singular points $z=0,1,\infty$. We have determined the monodromy properties of the solutions 
near the aforementioned singular points, and also found that our equations satisfies the so-called rigidity properties, which have eventually allowed us to conclude that exact global solutions 
are available and involve the generalized hypergeometric function  $ \,_4F_3$. For this hypergeometric function, a study of the Mellin-Barnes integral representation has allowed us to reach two 
fundamental goals: a complete analysis of the Stokes phenomenon and also a complete set of connection formulas, which are at the root of the description of the $S$-matrix for the 
scattering process associated with the analog Hawking effect.\\
On the other hand, we have taken into account some fundamental physical problems, which, of course, involve, as a focal point, the determination of the analog Hawking temperature. 
This part of the analysis has required some approximations, as fully analytical calculations are hard to be managed successfully. In particular, for the asymptotic expressions 
of the momenta of the modes involved in the scattering, we have adopted an expansion for low frequencies $\omega$, which is still standard in analytical calculations in literature. 
We have also considered both the subcritical regime and the transcritical one, and found explicit expressions for the Hawking temperature which are compatible both with 
the ones obtained in a perturbative framework in \cite{PaperoVariazionale} and, in the limit of weak dispersive effects, in \cite{master}. The aforementioned analysis, form a physical point 
of view, is just a very interesting but still incomplete one, as other regimes (beyond the low frequency one) can be investigated, and further amplitudes can be calculated, for a complete 
description of the full scattering matrix involved in the problem. We deserve a depeening and an extension of our study to future investigations.

\section*{Acknowledgements}
The authors are grateful to professor Yoshishige Haraoka for his precious advice and insight in  the mathematical aspects that are related to this work, especially concernng rigid Fuchsian equations.

\appendix

\section{Useful relations} \label{useful-relations}
We write (\ref{Eq_CharEq_inf}) and (\ref{Eq_CharEq_zero}) as
\begin{align} \label{Eq_DR_prod}
\text{DR}(k)&= \gamma^4 V^4 (k-  k_1)(k- k_2)(k-k_3)(k- k_4) \,, \\ \label{Eq_DR0_prod}
\text{DR}_0(k) &= \gamma^4 V^4 (k-\tilde k_1)(k-\tilde k_2)(k-\tilde k_3)(k-\tilde k_4) \,,
\end{align} 
where $k_j$ is a solution of $\text{DR}(k)=0$ and $\tilde k_j$ is a solutions of $\text{DR}_0(\tilde k)=0$.
By confronting (\ref{Eq_CharEq_inf}) with (\ref{Eq_DR_prod}) we deduce the following useful relations,
\begin{align}
 \frac 1 {(2\beta)^4}  k_1  k_2  k_3  k_4 &= \frac{\Omega^2 \left(-G^2 V^2+M^2+\Omega^2 \gamma^2\right)}{V^4 \gamma^2}\,, \\
 \frac 1 {(2\beta)^3} \sum_{i\neq j \neq l}  k_i  k_j  k_l  &= -\frac{2 \Omega \left(-G^2+M^2+2 \Omega^2 \gamma^2\right)}{V^3 \gamma^2}\,,\\
\frac 1 {(2\beta)^2} \sum_{i\neq j }  k_i k_j   &= \frac{V^2 \left(M^2+6 \Omega^2 \gamma^2\right)-G^2}{V^4 \gamma^2}\,, \\
 \frac 1 {(2\beta)}\sum_{i}  k_i  & = -\frac{4 \Omega}{V} \,, \label{Eq_sumK}
\end{align}
and similarly from (\ref{Eq_CharEq_zero}) and (\ref{Eq_DR0_prod})
\begin{align}
 \frac 1 {(2\beta)^4}\tilde k_1 \tilde k_2 \tilde k_3 \tilde k_4 &= \frac{\Omega^2 \left(-G^2 V^2+\Lambda+M^2+\Omega^2 \gamma^2\right)}{V^4 \gamma^2}\,,\\
\frac 1 {(2\beta)^3} \sum_{i\neq j \neq l} \tilde k_i \tilde k_j \tilde k_l  &= -\frac{2 \Omega \left(-G^2+\Lambda+M^2+2 \Omega^2 \gamma^2\right)}{V^3 \gamma^2}\,,\\
\frac 1 {(2\beta)^2} \sum_{i\neq j } \tilde k_i \tilde k_j   &= \frac{V^2 \left(\Lambda+M^2+6 \Omega^2 \gamma^2\right)-G^2}{V^4 \gamma^2}\,, \\
 \frac 1 {(2\beta)} \sum_{i} \tilde k_i   &= -\frac{4 \Omega}{V} \label{Eq_sumJ} \,.
\end{align}


\section{Proof of Proposition 1} \label{proposition-proof}
We prove the Theorem by induction. It is easy to verify that Eq. (\ref{Eq_Cn}) holds for $n=1,2,3,4$: indeed, substituting (\ref{Eq_SolutionExpansion}) into (\ref{Eq_FuchsianEQ_DR}) and truncating at order $4$ (using $\text{DR}_0(\tilde k_1)=0$) gives
	\begin{align}
	0=&z (-c_1 \text{DR}_0(\tilde k_1-i)+\text{DR}(\tilde k_1-i)) \cr 
	+&z^2 \big(c_1
	\text{DR}(\tilde k_1-2 i)+3 c_1 \text{DR}_0(\tilde k_1-i)-c_2 \text{DR}_0(\tilde k_1-2 i)-3 \text{DR}(\tilde k_1-i)\big) 
	\cr +&z^3 \Big(-3 c_1 (\text{DR}(\tilde k_1-2 i)+\text{DR}_0(\tilde k_1-i))+c_2 \text{DR}(\tilde k_1-3 i) \cr 
	&+3 c_2
	\text{DR}_0(\tilde k_1-2 i)-c_3 \text{DR}_0(\tilde k_1-3 i)+3 \text{DR}(\tilde k_1-i)\Big)\cr 
	+&z^4 \Big(3 c_1
	\text{DR}(\tilde k_1-2 i)+c_1 \text{DR}_0(\tilde k_1-i)-3 c_2 (\text{DR}(\tilde k_1-3 i)+\text{DR}_0(\tilde k_1-2 i))\cr 
	&+c_3
	\text{DR}(\tilde k_1-4 i)+3 c_3 \text{DR}_0(\tilde k_1-3 i)-c_4 \text{DR}_0(\tilde k_1-4 i)-\text{DR}(\tilde k_1-i)\Big)+O\left(z^5\right)\,,\cr
	\end{align}
	from which one can compute $c_1$, ..., $c_4$ by annihilating the coefficient of  each order. 
	Even though it is not necessary for the sake of the proof, we can verify (\ref{Eq_Cn}) also for some further $n$, by using the following identity, that is true for any fourth order polynomial\footnotemark $P_4(k)$,  \footnotetext{The identity holds more generally for any polynomial of order $k$, $$\sum_{m=0}^n \begin{pmatrix}n \\ m\end{pmatrix}(-1)^m P_k(m)=0 \,, \quad n > k\,,$$ and it derives immediately from the following property of binomial coefficients: $$ \sum_{m=0}^n \begin{pmatrix}n \\ m\end{pmatrix}(-1)^m m^k=0 \,, \quad n > k\,. $$}
	\begin{align}\label{Eq_RelationPolyn}
	0 = \sum_{m=0}^n \begin{pmatrix}n \\ m\end{pmatrix}(-1)^m P_4(-im) \,, \quad n > 4\,.
	\end{align} 
	The identity (\ref{Eq_RelationPolyn}) allows to write $\text{DR}(\tilde k_1-in)$ ($n\geq 5$) as a linear combination of $\text{DR}(\tilde k_1-i)$, $\text{DR}(\tilde k_1-2i)$, $\text{DR}(\tilde k_1-3i)$ and $\text{DR}(\tilde k_1-4i)$, and similarly  for $\text{DR}_0$.\\
	Now, for $n$ generic, assume that $c_n$, $c_{n+1}$, $c_{n+2}$, $c_{n+3}$ satisfy (\ref{Eq_Cn}). Take any two fourth degree polynomials
	\begin{align*}
	\text{DR}(k)&= a_0  + a_1 k + a_2 k^2 +a_3 k^3 +a_4 k^4 \,,\\
	\text{DR}_0(k)&= b_0  + b_1 k + b_2 k^2 +b_3 k^3 +b_4 k^4\,.
	\end{align*}
	Substituting (\ref{Eq_SolutionExpansion}) into (\ref{Eq_FuchsianEQ_DR}), we find that the coefficient $c_{n+4}$ satisfies the recurrence relation
	{\small 
		\begin{align}
		A_0\, c_{n} + A_1\, c_{n+1} + A_2\, c_{n+2}+ A_3\, c_{n+3}+A_4\, c_{n+4}=0\,,\label{Eq_Recurrence}
		\end{align}
		where
		\begin{align}
		&A_0 = \left(b_0-i b_1 (n+1)-b_2 (n+1)^2+i b_3 (n+1)^3+b_4 (n+1)^4\right)\cr 
		&+ k_1 \left(b_1-2 i b_2
		(n+1)-3 b_3 (n+1)^2+4 i b_4 (n+1)^3\right)\cr 
		& + k_1^2 \left(b_2-3 i b_3 (n+1)-6 b_4 (n+1)^2\right)\cr 
		&+
		k_1^3 (b_3-4 i b_4 (n+1))+ b_4 k_1^4 \,,
		\end{align}
		\begin{align}
		&A_1 = - \Big(a_0-i a_1 (n+1)-a_2 n^2-2 a_2 n-a_2+i a_3 n^3+3 i a_3 n^2+3 i a_3 n+i
		a_3+a_4 n^4\cr 
		&+4 a_4 n^3+6 a_4 n^2+4 a_4 n+a_4+3 b_0-3 i b_1 n-6 i b_1-3 b_2 n^2-12
		b_2 n-12 b_2+3 i b_3 n^3\cr 
		&+18 i b_3 n^2+36 i b_3 n+24 i b_3+3 b_4 n^4+24 b_4 n^3+72 b_4
		n^2+96 b_4 n+48 b_4\Big)\cr 
		&- k_1 \Big(a_1-2 i a_2 n-2 i a_2-3 a_3 n^2-6 a_3 n-3
		a_3+4 i a_4 n^3+12 i a_4 n^2+12 i a_4 n+4 i a_4\cr 
		&+3 b_1-6 i b_2 n-12 i b_2-9 b_3 n^2-36
		b_3 n-36 b_3+12 i b_4 n^3+72 i b_4 n^2+144 i b_4 n+96 i b_4\Big)\cr &- k_1^2 \left(a_2-3 i
		a_3 n-3 i a_3-6 a_4 n^2-12 a_4 n-6 a_4+3 b_2-9 i b_3 n-18 i b_3-18 b_4 n^2-72
		b_4 n-72 b_4\right)\cr 
		&- k_1^3 (a_3-4 i a_4 n-4 i a_4+3 b_3-12 i b_4 n-24 i b_4)-
		k_1^4 (a_4+3 b_4)\,,
		\end{align}
		\begin{align}
		&A_2= 3 \Big(a_0-i a_1 (n+2)-a_2 n^2-4 a_2 n-4 a_2+i a_3 n^3+6 i a_3 n^2+12 i a_3 n+8 i
		a_3+a_4 n^4\cr 
		&+8 a_4 n^3+24 a_4 n^2+32 a_4 n+16 a_4+b_0-i b_1 n-3 i b_1-b_2 n^2-6
		b_2 n-9 b_2+i b_3 n^3+9 i b_3 n^2 \cr 
		&+27 i b_3 n+27 i b_3+b_4 n^4+12 b_4 n^3+54 b_4
		n^2+108 b_4 n+81 b_4\Big)\cr 
		&+3 k_1 \Big(a_1-2 i a_2 n-4 i a_2-3 a_3 n^2-12 a_3 n-12
		a_3+4 i a_4 n^3+24 i a_4 n^2+48 i a_4 n\cr 
		&+32 i a_4+b_1-2 i b_2 n-6 i b_2-3 b_3 n^2-18
		b_3 n-27 b_3+4 i b_4 n^3+36 i b_4 n^2+108 i b_4 n+108 i b_4\Big)\cr 
		&+3 k_1^2 \Big(a_2-3 i
		a_3 n-6 i a_3-6 a_4 n^2-24 a_4 n-24 a_4+b_2-3 i b_3 n-9 i b_3-6 b_4 n^2-36 b_4
		n-54 b_4\Big)\cr 
		&+3 k_1^3 (a_3-4 i a_4 n-8 i a_4+b_3-4 i b_4 n-12 i b_4)+3 k_1^4
		(a_4+b_4)\,,
		\end{align}
		\begin{align}
		&A_3 =-3 \Big(3 a_0-3 i a_1 (n+3)-3 a_2 n^2-18 a_2 n-27 a_2+3 i a_3 n^3+27 i a_3 n^2+81 i a_3 n+81
		i a_3\cr 
		&+3 a_4 n^4+36 a_4 n^3+162 a_4 n^2+324 a_4 n+243 a_4+b_0-i b_1 n-4 i b_1-b_2
		n^2-8 b_2 n-16 b_2\cr 
		&+i b_3 n^3+12 i b_3 n^2+48 i b_3 n+64 i b_3+b_4 n^4+16 b_4 n^3+96
		b_4 n^2+256 b_4 n+256 b_4\Big)\cr &- k_1 \Big(3 a_1-6 i a_2 n-18 i a_2-9 a_3 n^2-54
		a_3 n-81 a_3+12 i a_4 n^3+108 i a_4 n^2+324 i a_4 n\cr 
		&+324 i a_4+b_1-2 i b_2 n-8 i b_2-3
		b_3 n^2-24 b_3 n-48 b_3+4 i b_4 n^3+48 i b_4 n^2+192 i b_4 n+256 i b_4\Big)\cr 
		&- k_1^2
		\Big(3 a_2-9 i a_3 n-27 i a_3-18 a_4 n^2-108 a_4 n-162 a_4+b_2-3 i b_3 n-12 i b_3\cr &-6
		b_4 n^2-48 b_4 n-96 b_4\Big)\cr 
		&- k_1^3 (3 a_3-12 i a_4 n-36 i a_4+b_3-4 i b_4 n-16 i
		b_4)- k_1^4 (3 a_4+b_4) \,,
		\end{align}
		\begin{align}
		& A_4 =  \left(a_0-i a_1 (n+4)-a_2 (n+4)^2+i a_3 (n+4)^3+a_4 (n+4)^4\right)\cr &+ k_1 \left(a_1-2 i a_2
		(n+4)-3 a_3 (n+4)^2+4 i a_4 (n+4)^3\right)\cr &+ k_1^2 \left(a_2-3 i a_3 (n+4)-6 a_4 (n+4)^2\right)+
		k_1^3 (a_3-4 i a_4 (n+4))+ a_4 k_1^4 \,.
		\end{align}
		Substituting the expressions for $c_{n}$, ...,  $c_{n+3}$ into Eq. (\ref{Eq_Recurrence}), we find that $c_{n+4}$ satisfies (\ref{Eq_Cn}) if and only if
		\begin{align}
		&A_0 \,\text{DR}_0(\tilde k_1-i (n+1)) \text{DR}_0(\tilde k_1-i (n+2)) \text{DR}_0(\tilde k_1-i (n+3)) \text{DR}_0(\tilde k_1-i (n+4))\cr&+A_1\,
		\text{DR}(\tilde k_1-i (n+1)) \text{DR}_0(\tilde k_1-i (n+2)) \text{DR}_0(\tilde k_1-i (n+3)) \text{DR}_0(\tilde k_1-i (n+4))\cr&+A_2\,
		\text{DR}(\tilde k_1-i (n+1)) \text{DR}(\tilde k_1-i (n+2)) \text{DR}_0(\tilde k_1-i (n+3)) \text{DR}_0(\tilde k_1-i (n+4))\cr&+A_3\,
		\text{DR}(\tilde k_1-i (n+1)) \text{DR}(\tilde k_1-i (n+2)) \text{DR}(\tilde k_1-i (n+3)) \text{DR}_0(\tilde k_1-i (n+4))\cr&+A_4\,
		\text{DR}(\tilde k_1-i (n+1)) \text{DR}(\tilde k_1-i (n+2)) \text{DR}(\tilde k_1-i (n+3)) \text{DR}(k_1-i (n+4))=0\,,\cr
		\end{align}}
	which indeed is true for any $n$, as can be checked by direct algebra or using Wolfram Mathematica. $\square$


\section{Physical consequences of assumption  \eqref{profile-os}} \label{refractive-sec}

We define, as in the non-dispersive case and in the weakly dispersive one, the refractive index to be 
\beq
n(x)=\sqrt{\chi(x)+1}.
\label{ref-index}
\eeq
Also, it is useful to define $n_0^2:=\chi_0+1$. We can rewrite \eqref{profile-os} as follows:
\beq
n^2 (x)-1=\frac{n_0^2-1}{1-\frac{\lambda}{2}(n_0^2-1)(1-\tanh (\beta x))}.
\label{squax}
\eeq
We obtain
\begin{align}
\lim_{x\to+\infty} (n^2 (x)-1) &=n_0^2-1, \label{cond-inf}\\
\lim_{x\to-\infty} (n^2 (x)-1) &=\frac{n_0^2-1}{1-\lambda (n_0^2-1)}. \label{cond-lam}
\end{align}
In standard materials we expect $n^2>1$. As a consequence, \eqref{cond-inf} implies $n^2>1$ as $x\to +\infty$  for $n_0^2>1$. The same request leads to 
$1-\lambda (n_0^2-1)>0$, which means $\lambda <\lambda_{\text{sup}}$, where 
\beq
\lambda_{\text{sup}}:=\frac{1}{n_0^2-1}=\frac{1}{\chi_0}.
\eeq
We can also wonder if we are assuming a black hole condition (decreasing $n(x)$) or a white hole one (increasing $n(x)$). Cf. e.g. \cite{master}. 
Given our monotone profile, we find that, by assuming, as we did, $\lambda>0$ we obtain a black hole geometry. A white hole one would be allowed by a negative $\lambda$. 
It is also to be noted that it is possible to satisfy both $\lambda <\lambda_{\text{sup}}$ and $\lambda \gg \lambda_{\text{crit}}$, as in the discussion following 
\eqref{temp-cr}, indeed we have 
\beq
\lambda_{\text{crit}}=\lambda_{\text{sup}} \gamma^2 (1-n_0^2 \frac{v^2}{c^2}),
\eeq
so that, for $n_0^2$ very near $\frac{c^2}{v^2}-\delta$, for $0<\delta \ll 1$, we get  $\lambda_{\text{sup}} \gg \lambda_{\text{crit}}$, and then 
$\lambda_{\text{sup}}>\lambda \gg \lambda_{\text{crit}}$ is allowed.


\begin{thebibliography}{10}

\bibitem{unruh} W.G. Unruh, 
Phys. Rev. Lett. \textbf{46}, 1351  (1981).

\bibitem{brout}
R.Brout, S.Massar, R.Parentani and Ph. Spindel, 
Phys. Rev. D, 52, 4559 (1995).

\bibitem{corley}
S. Corley, 
Phys. Rev. D, 57, 6280 (1998). 

\bibitem{himemoto}
Y. Himemoto,  and T. Tanaka, 
Phys. Rev. D, 61, 064004 (2000). 

\bibitem{saida}
H. Saida, and M. Sakagami, 
Phys. Rev. D, 61, 084023 (2000). 


\bibitem{Schutzhold-Unruh} 
  R.~Schutzhold and W.~G.~Unruh,
  Phys.\ Rev.\ D {\bf 78}, 041504 (2008).

\bibitem{unruh-s}

W.G. Unruh, and  R. Sch\"{u}tzhold, 
Phys. Rev. D 71,
024028  (2005). 

\bibitem{balbinot} 
  R.~Balbinot, A.~Fabbri, S.~Fagnocchi and R.~Parentani,
  Riv.\ Nuovo Cim.\  {\bf 28}, 1 (2005).


\bibitem{cpf} 
  A.~Coutant, R.~Parentani and S.~Finazzi,
  Phys.\ Rev.\ D {\bf 85}, 024021 (2012).

\bibitem{Leonhardt-Robertson} 
  U.~Leonhardt and S.~Robertson,
  New J.\ Phys.\  {\bf 14}, 053003 (2012).

\bibitem{Coutant-prd} 
  A.~Coutant, A.~Fabbri, R.~Parentani, R.~Balbinot and P.~Anderson,
  Phys.\ Rev.\ D {\bf 86}, 064022 (2012).

\bibitem{Coutant-und} 
  A.~Coutant and R.~Parentani,
   Phys.\ Fluids {\bf 26}, 044106 (2014).

\bibitem{Un-schu}
  R.~Schutzhold and W.~G.~Unruh,
  Phys.\ Rev.\ D {\bf 88}, 124009 (2013).

\bibitem{petev} 
  M.~Petev, N.~Westerberg, D.~Moss, E.~Rubino, C.~Rimoldi, S.~L.~Cacciatori, F.~Belgiorno and D.~Faccio,
  Phys.\ Rev.\ Lett.\  {\bf 111}, 043902 (2013).

\bibitem{Coutant-thick} 
  
A.~Coutant and R.~Parentani,
  Phys.\ Rev.\ D {\bf 90}, no. 12, 121501 (2014).
 
\bibitem{hopfield-hawking} 
  F.~Belgiorno, S.~L.~Cacciatori and F.~Dalla Piazza,
  Phys.\ Rev.\ D {\bf 91}, no. 12, 124063 (2015).

\bibitem{linder} 
  M.~F.~Linder, R.~Schutzhold and W.~G.~Unruh,
  Phys.\ Rev.\ D {\bf 93}, no. 10, 104010 (2016).

\bibitem{philbin-exact} 
  T.~G.~Philbin,
  Phys.\ Rev.\ D {\bf 94}, no. 6, 064053 (2016).

\bibitem{hopfield-kerr} 
  F.~Belgiorno, S.~L.~Cacciatori, F.~Dalla Piazza and M.~Doronzo,
  Phys.\ Rev.\ D {\bf 96}, no. 9, 096024 (2017).

\bibitem{CoutantWeinfurtner} 
  A.~Coutant and S.~Weinfurtner,
  Phys.\ Rev.\ D {\bf 94}, no. 6, 064026 (2016).

\bibitem{coutant-kdv} 
  A.~Coutant and S.~Weinfurtner,
  Phys.\ Rev.\ D {\bf 97}, no. 2, 025005 (2018).

\bibitem{coutant-bdg} 
  A.~Coutant and S.~Weinfurtner,
  Phys.\ Rev.\ D {\bf 97}, no. 2, 025006 (2018).



	\bibitem{Konig}
M. J. Jacquet, F. K\"onig,  Physical Review A, \textbf{102}(1), 013725 (2020).

\bibitem{rousseaux-first} 
  G.~Rousseaux, C.~Mathis, P.~Maissa, T.~G.~Philbin and U.~Leonhardt,
  New J.\ Phys.\  {\bf 10}, 053015 (2008).

\bibitem{faccio-prl} 
  F.~Belgiorno, S.~L.~Cacciatori, M.~Clerici, V.~Gorini, G.~Ortenzi, L.~Rizzi, E.~Rubino and V.~G.~Sala {\it et al.},
  Phys.\ Rev.\ Lett.\  {\bf 105}, 203901 (2010).

\bibitem{rubino-njp} 
  E.~Rubino, F.~Belgiorno, S.~L.~Cacciatori, M.~Clerici, V.~Gorini, G.~Ortenzi, L.~Rizzi and V.~G.~Sala {\it et al.},
  New J.\ Phys.\  {\bf 13}, 085005 (2011).

\bibitem{weinfurtner-prl} 
  S.~Weinfurtner, E.~W.~Tedford, M.~C.~J.~Penrice, W.~G.~Unruh and G.~A.~Lawrence,
  Phys.\ Rev.\ Lett.\  {\bf 106}, 021302 (2011).

\bibitem{rousseaux-book} 
  J.~Chaline, G.~Jannes, P.~Maissa and G.~Rousseaux,
  Lect.\ Notes Phys.\  {\bf 870}, 145 (2013).


\bibitem{weinfurtner-book}
  S.~Weinfurtner, E.~W.~Tedford, M.~C.~J.~Penrice, W.~G.~Unruh and G.~A.~Lawrence,
  Lect.\ Notes Phys.\  {\bf 870}, 167 (2013).

\bibitem{jeff-nature} 
  J.~Steinhauer,
  Nature Phys.\  {\bf 10}, 864 (2014).


\bibitem{rousseaux} 
  L.-P.~Euv\'e, F.~Michel, R.~Parentani, T.~G.~Philbin and G.~Rousseaux,
  Phys.\ Rev.\ Lett.\  {\bf 117}, no. 12, 121301 (2016).

\bibitem{denova} 
  J.~R.~Mu\~{n}oz de Nova, K.~Golubkov, V.~I.~Kolobov and J.~Steinhauer,
  Nature {\bf 569}, no. 7758, 688 (2019).

\bibitem{drori}
J.~Drori, Y.~Rosenberg, D.~Bermudez, Y.~Silberberg, and U.~Leonhardt,
Phys. Rev. Lett. {\bf 122}, 010404 (2019).


\bibitem{rousseaux-cocurrent} 
  L.~P.~Euv\`e, S.~Robertson, N.~James, A.~Fabbri and G.~Rousseaux,
  Phys.\ Rev.\ Lett.\  {\bf 124}, no. 14, 141101 (2020).

\bibitem{fourdinoy-weakly}
J.~Fourdrinoy, S.~Robertson, N.~James, A.~Fabbri and G.~Rousseaux,
Phys. Rev. D \textbf{105}, no.8, 085022 (2022).


\bibitem{master} 
  F.~Belgiorno, S.~L.~Cacciatori and A.~Vigan\`o,
  Phys.\ Rev.\ D {\bf 102}, no. 10, 105003 (2020).

\bibitem{bec} 
  F.~Belgiorno, S.~L.~Cacciatori, A.~Farahat and A.~Vigan\`o,
  Phys.\ Rev.\ D {\bf 102}, no. 10, 105004 (2020).

\bibitem{belcaccia-universe}
F. Belgiorno and S. L. Cacciatori, 
UNIVERSE 6, 127 (2020).
%


	
	
	\bibitem{PaperoVariazionale}
	S. Trevisan, F. Belgiorno and S.L. Cacciatori, 
	Phys. Rev. D \textbf{108}, 025001 (2023).
	
	
	
	\bibitem{Parentani_numerical}Michel, Parentani, 
	Phys. Rev.  D \textbf{90}, 044033  (2014).
	
	\bibitem{Haraoka} Y. Haraoka, \textit{Linear Differential Equations in the Complex Domain}, Springer International Publishing, 2020.

	\bibitem{coddington} E.A.Coddington, N.Levinson, \textit{Theory of Linear Differential Equations.} McGraw Hill, New York (1955).
	\bibitem{forsythe4} A.R.Forsyth,\textit{The Theory of Differential Equations. Part III. Ordinary Differential Equations. Vol. IV.} Cambridge: at the University Press (1902).
	\bibitem{craig} T.Craig, \textit{A Treatise on Linear Differential Equations.} Vol. I. John Wiley and Sons, New York (1889).
	\bibitem{Castro} A .Castro, J.M. Lapan, A. Maloney, M.J. Rodriguez, 
	Class. Quantum Grav. \textbf{30} (2013) 165005.
	\bibitem{DaCunha} B.C. da Cunha, F. Novaes,  
	Journal of High Energy Physics volume \textbf{2015}, 144 (2015).
	\bibitem{Oshima} T. Oshima, \textit{Fractional Calculus of Weyl Algebra and Fuchsian Differential Equations}, MSJ
	Meomoirs, vol. 28 (Mathematical Society of Japan, Tokyo, 2012).
	
	\bibitem{aomoto} K. Aomoto, M. Kita, \textit{Theory of hypergeometric functions}, Springer Tokyo (2011).
	
	\bibitem{hopfield-covariant}
F.~Belgiorno, S.~L.~Cacciatori and F.~Dalla Piazza,
Phys. Scripta \textbf{91}, no.1, 015001 (2016).
	
	

\bibitem{finazzi-carusotto-pra13} 
  S.~Finazzi and I.~Carusotto,
  Phys.\ Rev.\ A {\bf 87}, 023803 (2013).



\end{thebibliography}
\end{document}